\documentclass[a4paper,american,numberwithinsect]{lipics-v2018-modified}

\usepackage{style/allLIPICS}

\usepackage{thmtools}
\usepackage{thm-restate}

\title{
    Bounded Context Switching for Valence Systems
}

\author{Roland Meyer}
{TU Braunschweig, Germany}
{roland.meyer@tu-braunschweig.de}
{https://orcid.org/0000-0001-8495-671X}
{}

\author{Sebastian Muskalla}
{TU Braunschweig, Germany}
{s.muskalla@tu-braunschweig.de}
{https://orcid.org/0000-0001-9195-7323}
{}

\author{Georg Zetzsche}
{IRIF (Universit\'{e} Paris-Diderot, CNRS), France}
{zetzsche@irif.fr}
{https://orcid.org/0000-0002-6421-4388}
{Supported by a fellowship of the Fondation Sciences Math\'{e}matiques de
  Paris and partially funded by the DeLTA project
  (ANR-16-CE40-0007).}

\subjclass{\\
    ,\\
    ,\\
}


\authorrunning{R. Meyer, S. Muskalla and G. Zetzsche}

\Copyright{Roland Meyer, Sebastian Muskalla and Georg Zetzsche}

\newcommand{\indep}{I}
\newcommand{\indeprel}{\mathrel{\indep}}

\newcommand{\ops}
{\calO}

\newcommand{\inc}[1]
{{#1}^+}

\newcommand{\dec}[1]
{{#1}^-}

\newcommand{\incdec}[1]
{{#1}^\pm}

\DeclareDocumentCommand \tracemonoid { g } 
{ 
    \IfNoValueTF {#1}
    { \T_{G} }
    { \T_{#1} }
}

\DeclareDocumentCommand \graphmonoid { g } 
{ 
    \IfNoValueTF {#1}
    { \M_{G} }
    { \M_{#1} }
}

\DeclareDocumentCommand \BCSREACH { g }
{
  \IfNoValueTF {#1}
  {\mathsf{BCSREACH}}
  {\mathsf{BCSREACH}(#1)}
}

\DeclareDocumentCommand \BCSRP { g }
{
  \IfNoValueTF {#1}
  {\mathsf{BCSRP}}
  {\mathsf{BCSRP}(#1)}
}

\DeclareDocumentCommand \BCSI { m m }
{
  \mathsf{BCSINT}(#1,#2)
}

\DeclareDocumentCommand \langof { o m }
{
  \IfNoValueTF {#1}
  {\alang(#2)}
  {\alang_{#1}(#2)}
}

\newcommand{\grapheq}[1]{[#1]_\M}

\newcommand{\graphneutral}{1_\M}

\newcommand{\cs}[1]{\mathit{cs}\left(#1\right)}

\newcommand{\saturation}[1]{\mathit{sat}\left(#1\right)}

\newcommand{\frto}{\mapsto_{\mathit{free}}}
\newcommand{\rto}{\mapsto_{\mathit{red}}}

\newcommand{\frtop}[1]{\mapstochar\xrightarrow{#1}_{\mathit{free}}}

\newcommand{\tored}{\to_{red}}
\newcommand{\cancelrel}{\mathrel{R_\pi}}

\newcommand{\ord}[1]{\leq_{#1}}

\newcommand{\upr}[2]{{#1}^{(#2)}}
\newcommand{\ltr}[2]{#1[#2]}

\newcommand{\todec}{\rightsquigarrow}
\newcommand{\todecsat}{\rightsquigarrow_{\mathit{sat}}}

\newcommand{\todecsatw}[1]{\overset{#1}{\todec}_{\mathit{sat}}}
\newcommand{\nfa}{N}
\newcommand{\nfaof}[1]{\nfa(#1)}
\newcommand{\alang}{\mathcal{L}}

\newcommand{\init}{\mathit{init}}
\newcommand{\fin}{\mathit{fin}}
\newcommand{\qinit}{q_{\init}}
\newcommand{\qfin}{q_{\fin}}
\newcommand{\syninv}[1]{\mathit{sinv}(#1)}
\newcommand{\frcancelaut}{(FRA1)}
\newcommand{\frswapaut}{(FRA2)}

\newcommand{\opsof}[1]{\ops(#1)}

\newcommand{\tonfafun}{\mathit{2nfa}}
\newcommand{\tonfa}[2]{\tonfafun(#1, #2)}

\newcommand{\clipping}{\clip (-0.5,-0.45) rectangle (1.7,1.45);}

\newcommand{\squarelength}{1.2}

\newcommand{\trianglelength}{\squarelength*1.1547} 

\newcommand{\graphstack}
{%
  \begin{tikzpicture}[node distance=0.3cm]
    \clipping
        \node (a) at (0,0) {$\bullet$};
        \node [below of=a] {$a$};
        \node (b) at (60:\trianglelength) {$\bullet$};
        \node [below of=b] {$b$};
        \node (c) at (\trianglelength,0) {$\bullet$};
        \node [below of=c] {$c$};
    \end{tikzpicture}
}
\newcommand{\graphmultipushdown}
{%
  \begin{tikzpicture}[node distance=0.3cm]
    \clipping
        \node (a) at (0,0) {$\bullet$};
        \node [left of=a] {$0_\ell$};
        \node (b) at (0,\squarelength) {$\bullet$};
        \node [left of=b] {$1_\ell$};
        \node (c) at (\squarelength,0) {$\bullet$};
        \node [right of=c] {$0_r$};
        \node (d) at (\squarelength,\squarelength) {$\bullet$};
        \node [right of=d] {$1_r$};
        \path [draw] (a.center) -- (c.center)
        (a.center) -- (d.center)
        (b.center) -- (c.center)
        (b.center) -- (d.center);
    \end{tikzpicture}
}
\newcommand{\graphvass}
{%
  \begin{tikzpicture}[node distance=0.3cm]
    \clipping
        \node (a) at (0,0) {$\bullet$};
        \node [left of=a] {$p_1$};
        \node (b) at (0,\squarelength) {$\bullet$};
        \node [left of=b] {$p_2$};
        \node (c) at (\squarelength,0) {$\bullet$};
        \node [right of=c] {$p_3$};
        \node (d) at (\squarelength,\squarelength) {$\bullet$};
        \node [right of=d] {$p_4$};
        \path [draw] (a.center) -- (c.center)
        (a.center) -- (d.center)
        (b.center) -- (c.center)
        (b.center) -- (d.center)
        (a.center) -- (b.center)
        (c.center) -- (d.center);
    \end{tikzpicture}
}
\newcommand{\graphintegervass}
{%
  \begin{tikzpicture}[node distance=0.3cm]
    \clipping
        \node (a) at (0,0) {$\bullet$};
        \node [below= -3pt of a] {$c_1$};
        \node (b) at (60:\trianglelength) {$\bullet$};
        \node [left= 0pt of b] {$c_2$};
        \node (c) at (\trianglelength,0) {$\bullet$};
        \node [below= -3pt of c] {$c_3$};
        \path [draw]
        (a.center) -- (b.center)
        (b.center) -- (c.center)
        (c.center) -- (a.center);
        \draw (a.center) ++(210:3pt) circle (3pt);
        \draw (b.center) ++(90:3pt) circle (3pt);
        \draw (c.center) ++(-30:3pt) circle (3pt);
    \end{tikzpicture}
}

\newcommand{\Cfour}{\mathsf{C4}}
\newcommand{\Pfour}{\mathsf{P4}}

\newcommand{\picpfour}[1]{
    \begin{tikzpicture}[every circle/.style={}, scale=#1]
    \fill (0,0) circle (2pt) node (a) {}    (1,0) circle (2pt) node (b)  {}   (2,0) circle (2pt) node (c) {}   (3,0) circle (2pt) node (d) {};
    \draw (a.center) -- (b.center) -- (c.center) -- (d.center);
    \end{tikzpicture}
}

\newcommand{\piccfour}[1]{
    \begin{tikzpicture}[every circle/.style={}, scale=#1]
    \fill (0,0) circle (2pt) node (aa) {}    +(1,0) circle (2pt) node (ab) {}      +(1,1) circle (2pt) node (ac) {}    +(0,1) circle (2pt) node (ad) {};
    \draw (aa.center) -- (ab.center);
    \draw (aa.center) -- (ad.center) -- (ac.center);
    \draw (ab.center) -- (ac.center);
    \end{tikzpicture}
}

\newcommand{\leadstow}[1]
{\stackrel{#1}{\leadsto}}

\begin{document}

\maketitle

\begin{abstract}
    We study valence systems, finite-control programs over infinite-state memories modeled in terms of graph monoids. 
    Our contribution is a notion of bounded context switching (BCS).  
    Valence systems generalize pushdowns, concurrent pushdowns, and Petri nets. 
    In these settings, our definition conservatively generalizes existing notions. 
    The main finding is that reachability within a bounded number of context switches is in $\NPTIME$, independent of the memory (the graph monoid). 
    Our proof is genuinely algebraic, and therefore contributes a new way to think about BCS.
    In addition, we exhibit a class of storage mechanisms for which BCS reachability belongs to $\PTIME$.
  \end{abstract}

\newpage

\section{Introduction} 
%
Bounded context switching (BCS) is an under-approximate verification technique typically applied to safety properties.  
It was introduced for concurrent and recursive programs~\cite{QadeerR05}. 
There, a context switch happens if one thread leaves the processor for another thread to be scheduled. 
The analysis explores the subset of computations where the number of context switches is bounded by a given constant. 
Empirically, it was found that safety violations occur within few context switches~\cite{MQ07,LuPSY}.  
Algorithmically, the complexity of the analysis drops from undecidable to $\NPTIME$~\cite{QadeerR05,EGP14}. 
The idea received considerable interest from both practice and theory, a detailed discussion of related work can be found below. 

Our contribution is a generalization of bounded context switching to programs operating over arbitrary memories. 
To be precise, we consider valence systems, finite-control programs equipped with a potentially infinite-state memory modeled as a monoid~\cite{DOsualdoMeyerZetzsche2016a,Zetzsche2016d,Zetzsche2016c}. 
In valence systems, both the data domain and the operations are represented by monoid elements, and an operation $o$ will change the current memory value $m$ to the product $m\cdot o$. 
Of course, the monoid has to be given in some representation. 

We consider so-called graph monoids that capture the  memories commonly found in programs, like stacks, counters, and tapes, but also combinations thereof. 
A graph monoid is represented by a graph. 
Each vertex is interpreted as a symbol (say $c$) on which the operations push~($\inc{c}$) and pop~($\dec{c}$) are defined.
A computation is a sequence of such operations. 
The edges of the graph define an independence relation among the symbols that is used to commute the corresponding operations in a computation. 
To give an example, if $c$ and $d$ are independent, the computation $\inc{d}.\inc{c}.\dec{d}$ acts on two counters $c$ and $d$ and yields the values $1$ and $0$, respectively. 
Pushdowns are represented by valence systems over graphs without edges and concurrent pushdowns by complete $m$-partite graphs (for $m$ stacks). Petri nets yield complete graphs, blind counter systems complete graphs with self-loops on all vertices.

Our definition of context switches concentrates on the memory and does not reference the control flow. 
This frees us from having to assume a notion of thread, and makes the analysis applicable to sequential programs as well. 
We define a context switch as two consecutive operations in a computation that act on different and independent (in the above sense) symbols. 
This conservatively generalizes existing notions and yields intuitive behavior where a notion of context switch is not defined.
When modeling concurrent pushdowns, a context switch indeed corresponds to switching the stack.
For Petri nets and blind counter systems, it means switching the counter. 
Note, however, that the restriction can be applied to all memories expressible in terms of graph monoids. 


Our main result shows that reachability within a bounded number of context switches is in $\NPTIME$, \emph{for all graph monoids}. 
The result requires a uniform representation for the computations over very different memories.
We prove that a computation can always be split into quadratically-many blocks (in the number of context switches) --  independent of the monoid.
These blocks behave like single operations in that they commute or form inverses (in the given monoid). 
With this decomposition result, we develop an automata-theoretic approach to checking reachability. 
A more elaborate explanation of the proof approach can be found in Section~\ref{Section:BCS}, where we have the required terminology at hand.

In addition, we investigate the precise complexity of the problem for individual graph monoids. While  there are graph monoids for which our problem is $\NPTIME$-complete (such as those corresponding to the setting of concurrent pushdowns), we show that for an important subclass, those induced by transitive forests, the problem can be solved in polynomial time. Moreover, we describe those graph monoids for which the problem is $\NLOGSPACE$-complete. 

Taking a step back, our approach provides the first algebraic view to context-bounded computation, and hence enriches the tool box so far containing graph-theoretic interpretations and logical encodings of computations. 
We elaborate on the related work.  


\subparagraph*{Related Work.}
There are two lines of work on BCS that are closely related to ours in that they apply to various memory structures.
Aiswarya~\cite{Aiswarya14} and Madhusudan and Parlato~\cite{MP11} define a graph-theoretic interpretation of computations 
that manipulate a potentially infinite memory. 
They restrict the analysis to computations where graph-based measures like the split-width or the tree-width are bounded, 
and obtain general decidability results by reductions to problems on tree automata. 
The graph interpretation has been applied to multi pushdowns~\cite{AGK12},  timed systems~\cite{TimedTA,AGKS17}, and has been generalized to controller synthesis~\cite{AGK14}. 
It also gives a clean formulation of existing restrictions and  uniformizes the corresponding analysis algorithms, in particular for~\cite{QadeerR05,TorreMP07,TMP08,TMP10,Heussner12}.
Different from under-approximations based on split- or tree-width, we are able to handle counters, even nested within stacks.
We cannot handle, however, the queues to which those technique apply. 
Indeed, our main result is $\NPTIME$-completeness whereas graph-based analyses may have a higher complexity. 
Our approach thus applies to an incomparable class of models. 
Moreover, it contributes an algebraic view to bounded computations that complements the graph-theoretic interpretation. 

The second line of related work are reductions of reachability under BCS to satisfiability in existential Presburger arithmetic~\cite{EGP14,HagueL12}. 
Hague and Lin propose an expressive model, concurrent pushdowns communicating via reversal-bounded counters. 
Their main result is $\NPTIME$-completeness, like in our setting.
The model does not admit the free combination of stacks and counters that we support. 
The way it is presented, we in turn do not handle reversal boundedness, where the counters may change as long as the mode (increasing/decreasing) does not switch too often. 
Our approach should be generalizable to reversal boundedness by replacing the emptiness test in the free automata reduction of Section~\ref{Section:Algorithm} by a satisfiability check, using~\cite{VSS05}. 
The details remain to be worked out. 
Besides providing an incomparable class of models, our approach complements the logical view to computations.   

Reductions to existential Presburger arithmetic often restrict the set of computations by an intersection with a bounded language~\cite{GS64}, like in~\cite{EGP14,AAMS15}. 
The importance of bounded languages for under-approximation has been observed by Ganty et al.~\cite{GMM10,EsparzaGM12}. 




Besides the above unifying approaches, there has been a body of work on generalizations of BCS, towards exploring a larger set of computations~\cite{TorreMP07,TorreN11,EQR11,Atig14,TI0TP15a,AAA17} and handling more expressive programming models \cite{TMP08,AtigBQ09,Heussner12,BE14}. 
An unconventional instantance of the former direction are restrictions to the network topology~\cite{AtigBT08}. 
As particularly relevant instantiations of the latter, the BCS under-approximation has been applied to programs operating on relaxed memories~\cite{ABP11,AABN17} and programs manipulating data bases~\cite{AAAMR16}.

The practical work on BCS concentrated on implementing fast context-bounded analyses.
Sequentialization techniques~\cite{QW04} were successful in bridging the gap between the parallel program at hand and the available tooling, which is often limited to sequential programs. 
The idea is to translate the BCS instance into a sequential safety verification problem.
The first sequentialization for BCS has been proposed in~\cite{LalR08}, \cite{TorreMP09} gave a lazy formulation, and \cite{BouajjaniEP11} a systematic study of when sequentialization can be achieved.
The approach now applies to full C-programs \cite{CSEQ15} and has won the concurrency track in the softare verification competition. 
Current work is on parallelizing the analysis by further restricting the interleavings and in this way obtaining instances that are easier to solve~\cite{NS0TP17}.

Also with the goal of parallelization, recent works study the multi-variate complexity of context-bounded analyses.
While \cite{EGP14,FurbachM0S15} focus on $\PTIME$ and $\NPTIME$, \cite{ChiniKKMS17} studies fixed-parameter tractability, and \cite{CMS18} the fine-grained complexity.
The goal of the latter work is to achieve an analysis of comlexity $2^k\mathit{poly}(n)$, with $k$ a parameter and $n$ the input size.
Ideally, this analysis could be performed by $2^k$ independent threads, each solving a poly-time problem.

Our results contribute to a line of work on valence systems over graph
monoids~\cite{Zetzsche2016c}.  They have previously been studied with
respect to elimination of silent transitions~\cite{Zetzsche2013a},
semi-linearity of Parikh images~\cite{BuckheisterZetzsche2013a},
decidability of unrestricted reachability~\cite{Zetzsche2017a}, and
decidability of first-order logic with
reachability~\cite{DOsualdoMeyerZetzsche2016a}. See
\cite{Zetzsche2016d} for a general overview.


\section{Valence Systems over Graph Monoids}
We introduce the basics on graph monoids and valence systems following~\cite{Zetzsche2016c}.
%
\subparagraph*{Graph Monoids.}
Let $G = (V, I)$ be an undirected graph, without parallel edges, but possibly with self-loops. 
This means $\indep \subseteq V \times V$, which we refer to as the \emph{independence relation}, is symmetric but neither necessarily reflexive nor necessarily anti-reflexive.  
We use infix notation and write $o_1 \indeprel o_2$ for $(o_1, o_2) \in \indep$. 

To understand how the graph induces a monoid (a memory), think of the nodes $o \in V$ as stack symbols or counters. 
To each symbol $o$, we associate two operations, a positive operation $\inc{o}$ that can be understood as \emph{push $o$} or \emph{increment $o$} and a negative operation $\dec{o}$,  \emph{pop $o$} or \emph{decrement~$o$}. 
We call $+$ and $-$ the polarity of the operation.
By $\incdec{o}$ we denote an arbitrary element from $\set{ \inc{o}, \dec{o}}$. 
Let $\ops = \Set{ \inc{o}, \dec{o} }{ o \in V }$ denote the set of all operations. 
We refer to sequences of operations from $\ops^*$ as computations. 
We lift the independence relation to operations by setting $\incdec{o_1} \indeprel \incdec{o_2}$ if $o_1 \indeprel o_2$. 
We also write $v_1 \indeprel v_2$ for $v_1, v_2\in\ops^*$ if the operations in the computations are pairwise independent, and similar for subsets of operations $\ops_1 \indeprel \ops_2$ with $\ops_1, \ops_2\subseteq \ops$. 

We obtain the monoid by factorizing the set of all computations.
The congruence will identify computations that order independent operations differently.
Moreover, it will implement that $\inc{o}$ followed by $\dec{o}$ should have no effect, like a push followed by a pop.
Formally, we define $\cong$ as the smallest congruence (with respect to concatenation) on $\ops^*$ containing $\incdec{o_1}.\incdec{o_2} \cong \incdec{o_2}.\incdec{o_1}$ for all $o_1 \indeprel o_2$ and $\inc{o}.\dec{o} \cong \varepsilon$ for all $o$.

The \emph{graph monoid for graph $G$} is $\graphmonoid = \ops^* /_{\cong}$. 
For a word $w \in \ops^*$, we use $\grapheq{w} \in \graphmonoid$ to denote its equivalence class. 
Multiplication is $\grapheq{u}\cdot \grapheq{v}=\grapheq{u.v}$, which is well-defined as $\cong$ is a congruence.
The neutral element of $\graphmonoid$ is the equivalence class of $\varepsilon$, $\graphneutral = \grapheq{\varepsilon}$.

Recall that an element $x$ of a monoid $M$ is called \emph{right-invertible} if there is $y \in M$ such that $x \cdot y = 1_M$.
    We lift this notation to $\ops^*$ by saying that $w \in \ops^*$ is \emph{right-invertible} if its equivalence class $\grapheq{w} \in \graphmonoid$ is.

\vspace*{-0.2cm}

\subparagraph*{Valence Systems.}
Given a graph $G$, a \emph{valence system} over the graph monoid $\graphmonoid$ is a pair $A = (Q,\to)$, where $Q$ is a finite set of control states and $\to \, \subseteq Q \times (\ops \dotcup \set{\varepsilon}) \times Q$ is a set of transitions. 
A transition $q_1\tow{x} q_2$ is labeled by an operation on the memory. 
We write $q_1\to q_2$ if the label is $\varepsilon$, indicating that no operation is executed. 
The size of $A$ is $\card{A}=\card{\to}$.  
We use $\opsof{A}$ to access the set of operations that label transitions in $A$. 

A \emph{configuration} of $A$ is a tuple $(q,w) \in Q \times \ops^*$ consisting of a control state and the sequence of storage operations that has been executed.
We will restrict ourselves to configurations where $w$ is right-invertible.
More precisely, in $(q,w)$ a transition $q_1 \tow{x} q_2$ is \emph{enabled} if $q=q_1$ and $w.x$ is right-invertible.
In this case, the transition leads to the configuration $(q_2,w.x)$, and we write $(q, w)\to (q_2, w.x)$.
A \emph{run} is a sequence of consecutive transitions.

This restriction to right-invertible configurations is justified by the definition of the \emph{reachability problem} for valence systems.
It asks, given a valence system with two states $\qinit,\qfin$, whether we can reach $\qfin$ with neutral memory from $\qinit$ with neutral memory, \ie whether there is a run from $(\qinit,\varepsilon)$ to $(\qfin,w)$ with $\grapheq{w} = \graphneutral$. 
To be able to reach such a configuration $(\qfin,w)$ from some configuration $(q,w')$, $w'$ has to be right-invertible.
%
\subparagraph{Examples.}
Figure~\ref{Figure} depicts various graphs.
The graph monoid of each of these graph models a commonly used storage mechanism, \ie it represents the behavior of the storage.
\vspace*{-2mm}
\begin{enumerate}[(a)]
    \item Valence systems for this graph are pushdown systems over the stack alphabet $\set{a,b,c}$.
    \item Valence systems for this graph can be seen as concurrent pushdown systems with two stacks, each over a binary alphabet.
    \item Petri nets \resp vector addition systems with four counters/places $p_1, p_2, p_3, p_4$ can be modeled as valence systems for this graph.
        Since the valence system labels transitions by single increments or decrements, the transition multiplicities are encoded in unary.
    \item Integer vector addition systems \resp  blind counter automata with counters $c_1, c_2, c_3$ (that may assume negative values) can be seen as valence systems for this graph.
\end{enumerate}

\begin{figure}[hb]
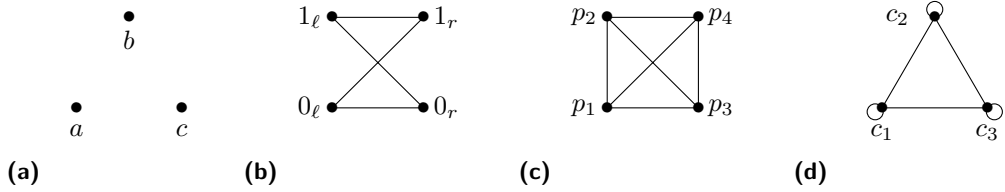

    {\centering\subcaptionbox{}[3cm][c]{\graphstack}}
    {\centering\subcaptionbox{}[3.5cm][c]{\graphmultipushdown}}
    {\centering\subcaptionbox{}[3.5cm][c]{\graphvass}}
    {\centering\subcaptionbox{}[3.5cm][c]{\graphintegervass}}
    \caption{Various examples of graphs representing commonly used storage mechanism.}
    \label{Figure}
\end{figure}

\subparagraph*{What about Queues?}

Let us quickly comment on why it is hard to fit queues into this
framework. An appealing aspect of valence automata over graph monoids
is that by using the monoid identity as the target for reachability
problems (\resp as an acceptance
condition~\cite{BuckheisterZetzsche2013a,Zetzsche2013a,Zetzsche2016c,Zetzsche2017a}),
we can realize a range of storage mechanisms by only varying the
underlying monoid. This is because in the mechanisms that we can
realize, the actions (or compositions of actions) that transform the
empty storage into the empty storage are precisely those that equal
the identity transformation.

In order to keep this aspect, we would need to construct a monoid
whose generators can be interpreted as queue actions so that a
sequence of generators transforms the empty queue into the empty queue
if and only if this sequence evaluates to the identity of the
monoid. This, however, is not possible: Suppose that $a$ and $b$
represent enqueue operations and that $\bar{a}$ and $\bar{b}$ are the
corresponding dequeue operations. Each of the action sequences $a.\bar{a}$
and $b.\bar{b}$ transforms the empty queue into the empty queue, but
$a.b.\bar{b}.\bar{a}$ does not (it is undefined on the empty
queue). Hence, in the monoid, we would want to have $a\bar{a}=1$,
$b\bar{b}=1$, but $ab\bar{b}\bar{a}\ne 1$, which violates
associativity. Hence, although it is possible to model queue behavior
in a monoid~\cite{HuschenbettKuskeZetzsche2017a,Kocher18,KocherK17},
one would need a different target element (or set).


\section{Bounded Context Switching}
\label{Section:BCS}

We introduce a notion of bounded context switching that applies to all valence systems, over arbitrary graph monoids. 
The idea is to let a new context start with an operation that is independent of the current computation, and hence intuitively belongs to a different thread. 
We elaborate on the notion of dependence. 

We call a set of symbols $V'\subseteq V$ \emph{dependent}, if it does not contain $o_1,o_2 \in V$, $o_1 \neq o_2$ with $o_1 \indeprel o_2$.
A set of operations $\ops'\subseteq \ops$ is dependent if its underlying set of symbols $\Set{o }{ \inc{o} \in \ops' \text{ or } \dec{o} \in \ops'}$ is.
A computation is dependent if it is over a dependent set of operations. 
A valence system is said to be dependent if the operations labeling the transitions form a dependent set. 
\begin{definition}
Given $w\in\ops^+$, its context decomposition is defined inductively:
If $w$ is dependent, $w$ is a single context and does not decompose.
Else, the first context $w_1$ of $w$ is the (non-empty) maximal dependent prefix of $w$.
Then, the context decomposition of $w$ is $w=w_1,\ldots,w_k$, where $w_2, \ldots, w_k$ is the context decomposition of the rest of the word.
The number of context switches in $w$, $\cs{w}$, is the number of contexts minus one. For technical reasons, it will be convenient to define $\cs{\varepsilon}=-1$.
\end{definition}
We study reachability under a restricted number of context switches.
\begin{problem}
    \problemtitle{Reachability under bounded context switching}
    \problemshort{(BCSREACH)}
    \probleminput{Valence system $A$, initial state $\qinit$, final state $\qfin$, bound $k$ in unary.}
    \problemquestion{Is there a run from $(\qinit,\varepsilon)$ to $(\qfin, w)$ so that $\grapheq{w} = \graphneutral$ and $\cs{w} \leq k$?}
\end{problem}

\noindent
In all abovementioned graph monoids, the restriction has an intuitive meaning that generalizes existing results.
Using the finite states, our notion of BCS
also permits a finite shared memory among the threads.
In addition, our definition applies to all storage structures expressible in terms of graph monoids, including combinations like stacks of counters.
\begin{lemma}
$\mathsf{(BCSREACH)}$ yields the following restriction:\vspace{-0.2cm}
\begin{enumerate}[(1)]
    \item On pushdowns, the notion does not incur a restriction.
    \item On concurrent pushdowns, the notion corresponds to changing the stack $k$-times and hence yields the original definition~\cite{QadeerR05}.
    \item On Petri nets and blind counters, the notion corresponds to changing the counter $k$-times.
\end{enumerate}
\end{lemma}

\noindent Our main result is this.
\begin{theorem}
\label{Theorem:NP}
    $\mathsf{(BCSREACH)}$ is in $\NPTIME$, independent of the storage graph.
\end{theorem}
%
Note that the $\NPTIME$ upper bound matches the lower bound in the
case of concurrent pushdowns~\cite{la2010language}.
We consider the proof technique the main contribution of the paper. 
Different from existing approaches, which are based on graph interpretations of computations or encodings into Presburger, ours is of algebraic nature. 
With an algebraic analysis, given in Section~\ref{Section:Decomposition}, we simplify the problem of checking whether a given computation reduces to one, $\grapheq{w}=\graphneutral$. 
We show that such a reduction exists if and only if the computation admits a decomposition into so-called blocks that reduce to one in a strong sense. 
There are two surprising aspects about the block decomposition. 
First, the strong reduction is defined by either commuting two blocks or canceling them if they are inverses. 
This means the blocks behave like operations, despite being full subcomputations. 
Second, the decomposition yields only 
quadratically-many blocks in the number of context switches (important for $\NPTIME$-membership).
The block decomposition is the main technical result of the paper.

The second step, presented in Section~\ref{Section:Algorithm}, is a symbolic check for whether a computation exists whose block decomposition admits a strong reduction.
We rely on automata-theoretic techniques to implement the operations of a strong reduction.
Key is a saturation based on which we give a complete check of whether two automata accept blocks that are inverses.


\section{Block Decomposition}
\label{Section:Decomposition}

In this section, we show how to decompose a computation that reduces to the neutral element into polynomially-many blocks such that the decomposition admits a syntactic reduction to $\varepsilon$.
The size of the decomposition will only depend on the number of contexts of the computation and not on its length.
This result will later provide the basis for our algorithm.

To be precise, we restrict ourselves to computations with so-called irreducible contexts.
In the next section, we will prove that the restriction to this setting is sufficient.

\begin{definition}
    We call a computation $w \in \ops^*$ \emph{irreducible} if it cannot be written as
    \( w = w'.a.w_I.b.w''\)
    such that $a = \inc{o}$, $b = \dec{o}$ and $o$ commutes with every symbol in $w_I$,
    or $a = \dec{o}, b = \inc{o}, o \indeprel o$ and $o$ commutes with every symbol in $w_I$.
\end{definition}
In other words, a computation is irreducible if we cannot eliminate a pair $\inc{o}.\dec{o}$ after using commutativity.
This is in fact the standard definition of irreducibility in the so-called trace monoid, which we do not introduce here.

Our goal is to decompose irreducible contexts such that the decomposition of all contexts in the computation admits a syntactic reduction defined as follows.

\newcommand{\frcancel}{(FR1)}
\newcommand{\frswap}{(FR2)}

\begin{definition}[\cite{LohreyZetzsche2017b}]
    Let $w_1, w_2, \ldots, w_n$ be a sequence of computations in $\ops^*$.
    A \emph{free reduction} is a finite sequence of applications of the following rewriting rules to consecutive entries of the sequence that transforms $w_1, \ldots, w_n$ into the empty sequence.
    \begin{enumerate}
        \item[\textbf{\frcancel}]
            $w_i, w_j \frto \varepsilon$\ ,
            applicable if $\grapheq{w_i.w_j} = \graphneutral$.
        \item[\textbf{\frswap}]
            $w_i, w_j \frto w_j, w_i$\ ,
            applicable if $w_i \indeprel w_j$
    \end{enumerate}

    \noindent
    We call $w_1, w_2, \ldots, w_n$ \emph{freely reducible} if it admits a free reduction.
\end{definition}
Being freely reducible is a strictly stronger property than 
$\grapheq{w_1 . w_2 . \ldots . w_n} = \graphneutral$:
It means that the sequence can be reduced to $\graphneutral$ by block-wise canceling, Rule~\frcancel, and swapping whole blocks, Rule~\frswap.
Indeed, consider $\inc{o_1}.\inc{o_2}, \dec{o_2}, \dec{o_1}$ where no two symbols commute.
We have \( \grapheq{\inc{o_1}.\inc{o_2}.\dec{o_2}.\dec{o_1}} = \graphneutral \), but the sequence is not freely reducible.

The decomposition of a computation $w$ with $\grapheq{w} = \graphneutral$ into its single operations is always freely reducible.
The main result of this section is that for a computation with irreducible contexts, we can always find a freely-reducible decomposition whose length is independent of the length of the computation.

\begin{theorem}
\label{Lemma:Crazy}
    Let $w$ be a computation with $\grapheq{w} = \graphneutral$ and let $w = w_1 \ldots w_k$ be its decomposition into irreducible contexts.
    There is a decomposition of each $w_i = w_{i,1}.w_{i,2}\ldots w_{i,m_i}$ such that
    $m_i \leq k - 1$ and the sequence
    \[
        w_{1,1},w_{1,2},\ldots,w_{1,m_1},
        w_{2,1},w_{2,2},\ldots,w_{2,m_2},
        \ldots,
        w_{k,1},w_{k,2},\ldots,w_{k,m_k}
    \]
    is freely reducible.
\end{theorem}
Note that the number of words occurring in the decomposition is bounded by $k^2$.
Theorem~\ref{Lemma:Crazy} can be seen as a strengthened version of Lemma~3.10 from~\cite{LohreyZetzsche2017b}:
We use the bound on the number of contexts to obtain a polynomial-size decomposition instead of an exponential one.
However, the proofs of the two results are vastly different.
\subparagraph{Constructing a Freely-Reducible Decomposition.}
The rest of this section will be dedicated to the proof of Theorem~\ref{Lemma:Crazy}.
Let $w \in \ops^*$ be the computation of interest with $\grapheq{w} = \graphneutral$.
We assume that it has length $n$ and $w=w_1 \ldots w_k$ is its decomposition into contexts.
For the first part of the proof, we do not require that each $w_i$ is irreducible.
As $\grapheq{w} = \graphneutral$, $w$ can be transformed into $\varepsilon$ by finitely often swapping letters and canceling out operations.
We formalize this by defining transition rules, similar to the definition of a free reduction.

For the technical development, it will be important to keep track of the original position of each operation in the computation.
To this end, we see $w$ as a word over $\ops \times \oneto{n}$, \ie we identify the \nth{$x$} operation $a$ of $w$ with the tuple $(a,x)$.
For ease of notation, we write $\ltr{w}{x}$ for the \nth{$x$} operation of $w$.
The annotation of letters by their original position will be preserved under the transition rules.

\newcommand{\rcancel}{(R1)}
\newcommand{\rcancelrev}{(R2)}
\newcommand{\rswap}{(R3)}

\begin{definition}
    A \emph{reduction} of $w$ is a finite sequence of applications of the following rewriting rules that transforms $w$ into into $\varepsilon$.
    \begin{enumerate}
        \item[\textbf{\rcancel}]
            \(
                w'.\ltr{w}{x}.\ltr{w}{y}.w''
                \rto
                w'.w''
                \ ,
            \)
            applicable if $\ltr{w}{x} = \inc{o}$, $\ltr{w}{y} = \dec{o}$ for some $o$.
        \item[\textbf{\rcancelrev}]
            \(
                w'.\ltr{w}{x}.\ltr{w}{y}.w''
                \rto
                w'.w''
                \ ,
            \)
            applicable if $\ltr{w}{x} = \dec{o}$, $\ltr{w}{y} = \inc{o}$ for $o\indeprel o$.

        \item[\textbf{\rswap}]
            \(
                w'.\ltr{w}{x}.\ltr{w}{y}.w''
                \rto
                w'.\ltr{w}{y}.\ltr{w}{x}.w''
                ,
            \)
            applicable if $\ltr{w}{x} \in \incdec{o_1}, \ltr{w}{y} \in \incdec{o_2}$ for $o_1 \indeprel o_2$,\\ \hspace*{2em} $o_1 \ne o_2$.
    \end{enumerate}

\end{definition}

\noindent
If a word $u$ can be transformed into $v$ using these rules, we write
$u\rto^* v$.  Note that a reduction of $w$ to $\varepsilon$ can be seen as a free
reduction of the sequence we obtain by decomposing $w$ into single
operations.  

\begin{restatable}{rlemma}{restateLemmaGedoens}
\label{Lemma:Gedoens}
    For a word $w$, we have $\grapheq{w} = \graphneutral$ iff $w$ admits a reduction.
\end{restatable}
\noindent
Consequently, we may fix a reduction $\pi = w \rto^* \varepsilon$ that transforms $w$ into $\varepsilon$.
The following definitions will depend on this fixed $\pi$.

\begin{definition}
    We define a relation $\cancelrel$ that relates positions of $w$ that cancel in $\pi$, \ie
    \[
        \ltr{w}{x} \cancelrel \ltr{w}{y}
        \quad \text{ if } \quad
            w'.\ltr{w}{x}.\ltr{w}{y}.w''
            \rto
            w'.w''
            \text{ or }
            w'.\ltr{w}{y}.\ltr{w}{x}.w''
            \rto
            w'.w''
        \text{ is used in } \pi
        \ .
    \]
    We lift it to infixes of $w$ by defining inductively
    \begin{align*}
        t_1 s_1 \cancelrel s_2 t_2
        \quad \text{ if } \quad
        & \text{there are contexts } w_i = w_{i1}.t_1.s_1.w_{i2}\text{ and }w_j = w_{j1}.s_2.t_2.w_{j2}
        \\
        &
        \text{of }w\text{ such that } s_1 \cancelrel s_2 \text{ and } t_1 \cancelrel t_2
        \ .
    \end{align*}
    An infix $u$ of a context $w_i$ is called a \emph{cluster} if
    there is an infix $u'$ of a context $w_j$ such that $u \cancelrel
    u'$. Moreover, if $u$ is a maximal cluster in $w_i$, then it is called a
    \emph{block}.  
\end{definition}
Note that $\cancelrel$ is symmetric by definition.
In the following, when we write $s_1 \cancelrel s_2$, we will assume that $s_1$ appears before $s_2$ in $w$, \ie $w = w'.s_1.w''.s_2.w'''$.
We now show that each context has a unique decomposition into blocks.
Afterwards, we will see that the resulting block decomposition is the decomposition required by Theorem~\ref{Lemma:Crazy}.

\begin{restatable}{rlemma}{restateLemmaUniqueFactorization}
    Every context has a unique factorization into blocks.
  \end{restatable}
\noindent
To prove the lemma, we show that each position belongs to at least one block and
to at most one block.
We call the unique factorization of a context $w_i$ into blocks the \emph{block decomposition} of $w_i$ (induced by $\pi$) and denote it by
\[
w_i = w_{i,1}, \ldots, w_{i,m_i}.
\]
The \emph{block decomposition} of $w$
(induced by $\pi$) is the concatenation of the block decompositions of its contexts,
\[
w = w_{1,1}, \ldots, w_{1,m_1}, \ldots, w_{k,1}, \ldots, w_{k,m_k} \ .
\]
Note that if $u$ is a block and $u\cancelrel v$, then $v$ is a block as well.
Therefore, $\cancelrel$ is a one-to-one correspondence of blocks.
It remains to prove that the block decomposition of $w$ admits a free reduction.
We will show that we can inductively cancel out blocks pairwise, starting with an \emph{innermost} pair.
Being innermost is formalized by the following relation.
\begin{definition}
    We define relation $\ord{w}$ on $\cancelrel$-related pairs of blocks by
    \(
        (s_1 \cancelrel s_2) \ord{w} (t_1 \cancelrel t_2)
    \)
    if
    \(
        w = \upr{w}{1} . t_1 . \upr{w}{2} . s_1 . \upr{w}{3} . s_2 .  \upr{w}{4} . t_2 . \upr{w}{5}
    \)
    for appropriately chosen $\upr{w}{1}, \ldots, \upr{w}{5}$.
    A pair $s_1 \cancelrel s_2$ minimal \wrt this order is called \emph{minimal nesting} in $w$.
\end{definition}
Note that we still assume that all letters are annotated by their position.
This means if $\upr{w}{1}, \ldots, \upr{w}{5}$ exist, they are uniquely determined.

\begin{restatable}{rlemma}{restateLemmaMinimalNesting}
\label{Lemma:MinimalNesting}
    $\ord{w}$ has a minimal nesting.
\end{restatable}
\noindent
The next lemma states that $s_1 \cancelrel s_2$ implies that $s_2$ is (a representative of) a right inverse of $s_1$. While we already know that the operations in $s_1$ cancel with those in $s_2$, it could ostensibly be the case that $\grapheq{s_2}$ is a left-inverse to $\grapheq{s_1}$.
\begin{restatable}{rlemma}{restateLemmaCancelrelInverses}
\label{Lemma:CancelrelInverses}
    If $s_1 \cancelrel s_2$, then $\grapheq{s_1 . s_2} = \graphneutral$.
\end{restatable}
\begin{restatable}{rproposition}{restatePropositionBlockDecompFreelyReducible}
\label{Prop:BlockDecompFreelyReducible}
    Let $\pi \colon w \tored^* \varepsilon$ be a reduction of $w$.
    The block decomposition of $w$ induced by $\pi$ is freely reducible.
\end{restatable}

\begin{proof}
    If $w = \varepsilon$, then there is nothing to do.
    Otherwise, $w$ decomposes into at least two blocks.
    We proceed by induction on the number of blocks.
%
    In the base case, let us assume that $w = u,v$ is the block decomposition, where $u \cancelrel v$ has to hold.
    Using Lemma~\ref{Lemma:CancelrelInverses},
    \(
        u,v \frtop{\text{\frcancel}} \varepsilon
    \)
    is the desired free reduction.
    
    In the inductive step, we pick a minimal nesting $s_1 \cancelrel s_2$ in $w$.
    As argued in Lemma~\ref{Lemma:MinimalNesting}, this is always possible.
    We may write
    \[
        w
        =
        w_1 \ldots
        \underbrace{w_{i_1} s_1 w_{i_2}}_{\text{context } w_i}
        \ldots
        \underbrace{w_{j_1} s_2 w_{j_2}}_{\text{context } w_j}
        \ldots w_k
        \ .
    \]
    Since $s_1 \cancelrel s_2$, we know that by definition of $\cancelrel$,
    $\pi$ has to move each letter from $s_1$ next to the corresponding letter of $s_2$ or vice versa.
    
    Let us consider the effect of $\pi$ on the infix $w_{i_2} \ldots w_{j_1}$.
    Without further arguments, the reduction $\pi$ could cancel some letters inside this infix, and it can swap the remaining letters with the letters in $s_1$ or $s_2$.
    In fact, there can be no canceling within $w_{i_2} \ldots w_{j_1}$, as $s_1 \cancelrel s_2$ was chosen to be a minimal nesting:
    Assume that $w_{i_2} \ldots w_{j_1}$ contains some letters $a,b$ with $a \cancelrel b$.
    Pick the unique blocks $u,v$ to which they belong, and note that we have $(u \cancelrel v) <_{w} (s_1 \cancelrel s_2)$, \ie $(u \cancelrel v) \ord{w} (s_1 \cancelrel s_2)$ and $(u,v) \neq (s_1,s_2)$, a contradiction to the minimality of $s_1 \cancelrel s_2$.
    
    Hence, the reductions needs to swap all letters in $w_{i_2} \ldots w_{j_1}$ with $s_1$ or $s_2$ and we have $s_1 \indeprel w_{i_2} \ldots w_{j_1} \indeprel s_2$.
    We construct a free reduction as follows:
    \begin{align*}
        &
        w_1 \ldots w_{i_1} s_1 w_{i_2} w_{i+1} \ldots w_{j-1} w_{j_1} s_2 w_{j_2} \ldots w_k
        \\
        \frtop{\text{\frswap}}^*\quad &
        w_1 \ldots w_{i_1} w_{i_2} w_{i+1} \ldots w_{j-1} w_{j_1} s_1 s_2 w_{j_2} \ldots w_k
        \\
        \frtop{\text{\frcancel}}\quad &
        w_1 \ldots w_{i_1} w_{i+1} \ldots w_{j-1} w_{j_2} \ldots w_k =: w'
        \ .
    \end{align*}
    The applications of Rule~\frswap\ are valid as $s_1 \indeprel w_{i_2} \ldots w_{j_1} \indeprel s_2$ holds.
    The application of Rule~\frcancel\ to $s_1,s_2$ is valid by Lemma~\ref{Lemma:CancelrelInverses}.
    
    Let us denote by $w'$ the result of these reduction steps.
    We consider the reduction $\pi'$ that is obtained by restricting $\pi$ to transitions that work on letters still present in $w'$.
    Indeed, $\pi'$ reduces $w'$ to $\varepsilon$.
    In particular, for each operation in $w'$, the operation it cancels with is the same in $\pi$ and $\pi'$.
    Consequently, the relation $R_{\pi'}$ is the restriction of $\cancelrel$ to the operation still occurring in $w'$, and the block decomposition of $w'$ induced by $\pi'$ is the block decomposition of $\pi$ minus the blocks $s_1, s_2$ that have been removed.
    
    We may apply induction to obtain that $w'$ admits a free reduction.
    We prepend the above reduction steps to this free reduction to obtain the desired reduction for $w$.
    
    We emphasize the fact that we have not used in the proof that the $w_i$ are contexts.
    This is important, as the context decompositions of $w$ and $w'$ can differ substantially.
    Potentially, we have that $w$ consists of four contexts, $w = w_1,s_1,w_2,s_2$, but after canceling $s_1$ with $s_2$, $w_1$ and $w_2$ merge to a single context, $w' = w_1.w_2$.
    As we have preserved $\cancelrel$ and its induced block decomposition, this does not hurt the validity of the proof.
\end{proof}
\subparagraph{A Bound on the Number of Blocks.}
It remains to prove the desired bound on the number of blocks.
To this end, we will exploit that each context $w_i$ is irreducible.
%
\begin{proposition}
\label{Prop:BlockDecompSize}
    Let $w$ be a computation with irreducible contexts and $\pi \colon w \tored^* \varepsilon$ a reduction.
    In the block decomposition of $w$ induced by $\pi$, $m_i \leq k-1$ holds for all $i$.
\end{proposition}
We prove the proposition in the form of two lemmas.
\begin{restatable}{rlemma}{restateLemmaCancelrelNotSameBlock}
\label{Lemma:CancelrelNotSameBlock}
    The relation $\cancelrel$ never relates blocks from the same context.
\end{restatable}
\noindent
%
%
%
The following lemma allows us to bound the number of blocks in a context by the total number $k$ of contexts.
\begin{restatable}{rlemma}{restateLemmaCancelrelNotTwice}
\label{Lemma:CancelrelNotTwice}
For any two contexts $w_i$ and $w_j$, there is at most one block in $w_i$ that is $\cancelrel$-related to a block in $w_j$.
\end{restatable}
\begin{proof}
    Towards a contradiction, assume that some context contains two blocks that are $\cancelrel$-related to a block from the same context.
    Let us consider the minimal $i$ such that $w_i$ contains such blocks.
    Let $w_j$ be the context to which the two blocks are related.
    By the choice of $i$, $w_i$ occurs in $w$ before $w_j$ does.
    
    We pick $s_1,t_1$ as a pair of blocks in $w_i$ canceling with blocks from $w_j$ with minimal distance, \ie $w_i = w_{i_1} s_1 w_{i_2} t_1 w_{i_3}$ where $w_{i_2}$ contains no block that is canceled by some block in $w_j$.
    Let $s_2, t_2$ be the blocks in $w_j$ such that $s_1 \cancelrel s_2$, $t_1 \cancelrel t_2$.
    We have to distinguish two cases, depending on the order of occurrence of $s_2$ and $t_2$ in $w_j$.
    In the first case, we have $w_j = w_{j_1} t_2 w_{j_2} s_2 w_{j_3}$ and thus
    \[
        w = w_1 \ldots w_{i-1}
        \ \underbrace{ w_{i_1} s_1 w_{i_2} t_1 w_{i_3} }_{\text{context } w_i} \ 
        w_{i+1} \ldots w_{j-1} 
        \ \underbrace{  w_{j_1} t_2 w_{j_2} s_2 w_{j_3} }_{\text{context } w_j} \ 
        w_{j+1} \ldots w_k
        \ .
    \]
    Our goal is to show that $w_{i_2}$ and $w_{j_2}$ have to be empty.
    We then obtain $s_1 t_1 \cancelrel t_2 s_2$, a contradiction to the definition of blocks as maximal $\cancelrel$-related infixes in each context.
    
    We start by assuming that $w_{i_2}$ contains some operation $b$.
    As $\pi$ reduces $w$ to $\varepsilon$, $w$ contains some operation $c$ that $b$ cancels with.
    We first note that $c$ cannot be contained in $w_{j}$, as we have chosen $s_1, t_1$ such that $w_{i_2}$ contains no block that cancels with a block of $w_j$.
    Assume that $c$ is contained in the prefix $w_1 \ldots w_{i-1} w_{i_1}$.  Reduction $\pi$ either
    needs to swap $b$ or $c$ with $s_1$, or it needs to swap $s_2$ with $b$ (to cancel $s_1$).  In
    any case, by definition of $\rto$, this means $s_1$ contains an operation that commutes with $b$
    and is distinct from $b$. However, this is impossible, as $s_1$ and $b$ are contained in the
    same context $w_i$, and contexts do not contain distinct independent symbols.
    For the same reason, $c$ cannot be contained in the suffix $w_{j_3} w_{j+1} \ldots w_k$.
    
    If $c$ is contained in the infix $w_{i+1} \ldots w_{j-1}$, $\pi$ needs to swap $b$ with $t_1$, or $c$ with $t_1$, or $t_2$ with $c$.
    In any case, this means $t_1$ contains an operation that commutes with $b$ and is distinct from $b$. However, this is impossible, as $t_1$ and $b$ are contained in the same context $w_i$, and contexts do not contain distinct independent symbols.
    
    Consequently $w_{i_2}$ needs to be empty.
    Let us assume that $w_{j_2}$ contains an operation $b$, and let $c$ denote the operation it cancels with.
    As for $w_{i_2}$, we can show that $c$ can neither be contained in the prefix $w_1 \ldots w_{i-1} w_{i_1}$, nor in the suffix $w_{j_3} w_{j+1} \ldots w_k$, nor in the infix $w_{i+1} \ldots w_{j-1}$.
    We conclude that $w_{j_2}$ is also empty and obtain a contradiction to the maximality of the blocks as explained above.
            
    It remains to consider the second case, \ie $w_j = w_{j_1} s_2 w_{j_2} t_2 w_{j_3}$ and
    \[
        w = w_1 \ldots w_{i-1}
        \ \underbrace{ w_{i_1} s_1 w_{i_2} t_1 w_{i_3} }_{\text{context } w_i} \ 
        w_{i+1} \ldots w_{j-1}
        \ \underbrace{  w_{j_1} s_2 w_{j_2} t_2 w_{j_3} }_{\text{context } w_j}\ 
        w_{j+1} \ldots w_k
        \ .
    \]
    Reduction $\pi$ either needs to swap $s_1$ with $t_1$ or equivalently $s_2$ with $t_1$.  Again
    by definition of $\rto$, this means there is an operation $a$ in $s_1$ and an operation $b$ in
    $t_1$ such that $a \indeprel b$ and $a,b$ have distinct symbols.
    Since $s_1, t_1$ and $s_2,t_2$ belong to the same context, this is impossible.
\end{proof}
Lemma~\ref{Lemma:CancelrelNotSameBlock} and Lemma~\ref{Lemma:CancelrelNotTwice} together prove Proposition~\ref{Prop:BlockDecompSize}, finishing the proof of Theorem~\ref{Lemma:Crazy}.


\section{Decision Procedure}
\label{Section:Algorithm}

Given a valence system $A$ with states $\qinit$ and $\qfin$, and a bound $k$, we give an algorithm that checks whether there is a run from $(\qinit,\varepsilon)$ to $(\qfin,w)$ such that $\grapheq{w} = \graphneutral$ and $\cs{w} \leq k$. 

\subparagraph*{Implementing Irreducibility.}
The theory we have developed above applies to irreducible contexts.
To determine the irreducible versions of contexts in $A$, we define a saturation operation on valence systems. 
The algebraic idea behind the saturation is the following.
\begin{lemma}\label{Lemma:SaturateDependent}
Let $w$ be a dependent computation. 
Then $w$ can be turned into an irreducible computation by applying the following rules: $\inc{o}.\dec{o}\mapsto\varepsilon$ and, provided $o\indeprel o$, $\dec{o}.\inc{o}\mapsto\varepsilon$.
\end{lemma}
To see the lemma, note that in a dependent computation, reducible operations $\inc{o}$ and $\dec{o}$ cannot be separated by an operation on a different symbol.  
Hence, $\inc{o}$ and $\dec{o}$ are placed side by side (potentially after further reductions).
If $o\indeprel o$ does not hold, the first rule is sufficient for the reduction.
If $o\indeprel o$ does holds, we may find $\dec{o}.\inc{o}$ and need both rules.

The saturation operation implements these two rules. 
Since Lemma~\ref{Lemma:SaturateDependent} assumes a dependent computation, we consider a dependent valence system $B = (P, \todec)$. 
The \emph{saturation} is the valence system $\saturation{B}=(P,\todecsat)$ with the same set of control states.
The transitions are defined by requiring $\todec\ \subseteq\ \todecsat$ and exhaustively applying the following rules:
\begin{enumerate}[(1)]
    \item If $p_1 \todecsatw{\inc{o}} p \todecsatw{}^* p' \todecsatw{\dec{o}} p_2$, add an $\varepsilon$-transition $p_1 \todecsatw{} p_2$.
    \item If $p_1 \todecsatw{\dec{o}} p \todecsatw{}^* p' \todecsatw{\inc{o}} p_2$ and $o \indeprel o$, add an $\varepsilon$-transition $p_1 \todecsatw{} p_2$.    
\end{enumerate}

\noindent
Here, $p \todecsatw{}^* p'$ denotes that $p'$ is reachable from $p$ by a sequence of $\varepsilon$-transitions.
\begin{remark}
    In the worst case, we add $\card{P}^2$ many transitions.
\end{remark}
\begin{restatable}{rlemma}{restateLemmaSaturation}
\label{Lemma:Saturation}
    There is a computation $(q_1,\varepsilon) \to (q_2, u)$ in $B$
    if and only if
    there is a computation $(q_1,\varepsilon) \to (q_2,v)$ with $v$ irreducible and $u\cong v$ in $\saturation{B}$.
\end{restatable}
\noindent
%
The valence system $A=(Q, \to)$ of interest may not be dependent. 
We will determine dependent versions of it (one for each context) by restricting to a dependent set of operations $\ops'\subseteq \ops$.
The \emph{restriction} is defined by $A[\ops']=(Q, \to\cap\ (Q\times (\ops'\cup\set{\varepsilon})\times Q))$. 
\subparagraph*{Representing Block Decompositions.}
Theorem~\ref{Lemma:Crazy} considers a computation decomposed into irreducible contexts $w_1$ to $w_k$.  
It shows that each context $w_i$ can be further decomposed into at most $k$ blocks such that the overall sequence of blocks
$w_{1, 1},\ldots, w_{k, m_k}$ freely reduces to $\graphneutral$. 
Our goal is to represent the block decompositions of all candidate computations in a finite way. 
To this end, we analyze the result more closely.

The decomposition into contexts means there are dependent sets $\ops_1,\ldots, \ops_k\subseteq\ops$ such that each context $w_i$ only uses operations from the set $\ops_i$. 
The decomposition into blocks means there are $n=k^2$ computations $v_1$ to $v_n$ and states $q_1$ to $q_{n-1}$ such that $v_i$ leads from $q_{i-1}$ to $q_i$ with $q_0=\qinit$ and $q_n=\qfin$. 
The last thing to note is that a block itself does not have to be right-invertible.  
This means we should represent block decompositions by (non-deterministic finite) automata rather than valence systems. 

We define, for each pair of states $q_{i}, q_{f}\in Q$, each dependent set of operations $\ops_{\mathit{con}}\subseteq \ops$, and each subset $\ops_{\mathit{bl}}\subseteq \ops_{\mathit{con}}$ the automaton 
\begin{align*}
\nfaof{q_{i}, q_{f}, \ops_{\mathit{con}}, \ops_{\mathit{bl}}} = \tonfa{q_i, q_f}{\saturation{A[\ops_{\mathit{con}}]}[\ops_{\mathit{bl}}]}\ .
\end{align*}
Function $\tonfafun$ understands the given valence system $\saturation{A[\ops_{\mathit{con}}]}[\ops_{\mathit{bl}}]$ as an automaton, with the first parameter as the initial and the second as the final state. 
The set $\ops_{\mathit{con}}$ will be the operations used in the context of interest. 
As these operations are dependent, $\saturation{A[\ops_{\mathit{con}}]}$ will include the irreducible versions of all computations in $A[\ops_{\mathit{con}}]$, Lemma~\ref{Lemma:Saturation}.
The second restriction to $\ops_{\mathit{bl}}$ identifies the operations of one block in the context.

With this construction at hand, we define our representation of block decompositions. 
\begin{definition}
A \emph{test} for the given $\mathsf{(BCSREACH)}$-instance is a sequence $\nfa_1,\ldots,\nfa_n$ of $n=k^2$ automata $\nfa_i=\nfaof{q_{i-1}, q_i, \ops_j, \ops_{j, i}}$ with $j=\lceil \frac{i}{k}\rceil$, $q_0=\qinit$, and $q_n=\qfin$.
\end{definition}
%
The following lemma links Theorem~\ref{Lemma:Crazy} and the notion of tests. 
With Theorem~\ref{Lemma:Crazy}, we have to check whether there is a computation $w$ from $\qinit$ to $\qfin$ with $\cs{w}\leq k$ whose block decomposition admits a free reduction. 
With the analysis above, such a computation exists iff there is a test $N_1$ to $N_n$ whose automata accept the blocks in the decomposition.
\begin{restatable}{rlemma}{restateLemmaReductionAutomataA}
\label{Lemma:ReductionAutomataA}
    We have $(\qinit,\varepsilon) \to (\qfin,w)$ with $\cs{w}\leq k$ and $\grapheq{w}=1$ in $A$ iff there is a test $\nfa_1,\ldots,\nfa_n$ and computations $v_1\in\langof{\nfa_1}$ to $v_n\in\langof{\nfa_n}$ that freely reduce to $\graphneutral$.
\end{restatable}
\subparagraph*{Determining Free Reducibility.}
Given a test $\nfa_1,\ldots, \nfa_n$, we have to check whether the automata accept computations that freely reduce to $\graphneutral$.  
To get rid of the reference to single computations, we now define a notion of free reduction directly on sequences of automata.
This means we have to lift the following operations from computations to automata.
On computations $u$ and $v$, a free reduction may check commutativity, $u \indeprel v$, and whether the computations are inverses, $\grapheq{u}\cdot \grapheq{v}=\graphneutral$. 
Consider $\nfa_u$ and $\nfa_{v}$ from $\nfa_1,\ldots,\nfa_n$. 

Rather than checking whether $\nfa_u$ and $\nfa_{v}$ accept computations that commute, the free reduction on automata will check whether the alphabets are independent, $\opsof{\nfa_{u}} \indeprel \opsof{\nfa_{v}}$. 
To see that this yields a complete procedure, note that Lemma~\ref{Lemma:ReductionAutomataA} existentially quantifies over all tests, and hence all sets of operations to construct $N_u$ and $N_v$. 
If there are computations $u$ and $v$ that commute in the free reduction, we can construct the automata $\nfa_u$ and $\nfa_v$ by restricting to the letters in these words. 
This will still guarantee $u\in\langof{\nfa_u}$ and $v\in\langof{\nfa_v}$.

To check whether $\nfa_u$ and $\nfa_v$ accept computations that multiply up to $\graphneutral$, we rely on the syntactic inverse. 
Consider a computation $u$ that contains negative operations $\dec{o}$ only for symbols with $o\indeprel o$. 
In this case, the \emph{syntactic inverse} $\syninv{u}$ is defined by reversing the letters and inverting the polarity of operations. 
The operation is not defined otherwise. 
The following lemma is immediate.
\begin{restatable}{rlemma}{restateLemmaSyntacticInverse}
\label{Lemma:SyntacticInverse}
    If $u, v\in\ops^*$  are irreducible, dependent with $\grapheq{u}\cdot\grapheq{v}=\graphneutral$, then $v=\syninv{u}$. 
\end{restatable}
\noindent
%
The idea is to admit $v$ as the inverse of $u$ if $v=\syninv{u}$ holds. 
The equality will of course entail that $v$ is the inverse of $u$, for any pair of computations. 
Lemma~\ref{Lemma:SyntacticInverse} moreover shows that for irreducible, dependent computations the check is complete. 
Since $\nfa_u$ and $\nfa_v$ are dependent and saturated, it will be complete (Lemma~\ref{Lemma:Saturation}) to use the syntactic inverse also on the level of automata.

The definition swaps initial and final state, turns around the transitions, removes the negative operations on non-commutative symbols, and inverts the polarity of the others. 
Formally, the \emph{syntactic inverse} yields $\syninv{N_u}=(Q, q_{u, \fin}, \mathit{remswap}(\to^{-1}_u), q_{u, \init})$. 
The reverse relation contains  
$(q_2, \incdec{o}, q_1)\in\ \to_u^{-1}$ iff $(q_1, \incdec{o}q_2)\in\ \to_u$. 
Function $\mathit{remswap}$ removes transitions with operations $\dec{o}$ for which $o \indeprel o$ does not hold and inverts the remaining polarities.
%
The construction guarantees that $\syninv{\langof{N_u}}=\langof{\syninv{N_u}}$. 
With this, the check of whether $N_u$ and $N_v$ contain computations $u$ and $v$ with $v=\syninv{u}$ amounts to checking whether $N_v$ and $\syninv{N_u}$ have a computation in common.
\begin{lemma}
There are $u\in\langof{N_u}, v\in\langof{N_v}$ with $v=\syninv{u}$ iff $\langof{N_v}\cap \langof{\syninv{N_u}}\neq \emptyset$.
\end{lemma}
%
The analogue of the free reduction defined on automata is the following definition.
\begin{definition}
A \emph{free automata reduction} on a test $N_1$ to $N_n$ is a sequence of operations 
    \begin{enumerate}
       \item[\textbf{\frcancelaut}]
            $N_i, N_{j} \frto \varepsilon$, if 
            $\langof{N_{j}}\cap \langof{\syninv{N_{i}}}\neq \emptyset$.
        \item[\textbf{\frswapaut}]
            $N_i, N_j \frto N_j, N_i$, if $\opsof{N_i}\indeprel\opsof{N_j}$.
    \end{enumerate}
\end{definition}

\noindent
Since we quantify over all tests, free automata reductions are complete as follows.
\begin{lemma}\label{Lemma:ReductionAutomataB}
There is a test $\nfa_1, \ldots, \nfa_n$ and computations $u_1\in\langof{\nfa_1}$ to $u_n\in\langof{\nfa_n}$ that freely reduce to $\graphneutral$ iff 
there is a test $\nfa_1, \ldots, \nfa_n$ that admits a free automata reduction to $\varepsilon$. 
\end{lemma}
Together, Lemma~\ref{Lemma:ReductionAutomataA} and Lemma~\ref{Lemma:ReductionAutomataB} yield a decision procedure for $\mathsf{(BCSREACH)}$. 
We guess a suitable test and for this test a suitable free automata reduction. 
The restrictions, the saturation, the automata conversion, and the independence and disjointness tests require time polynomial in $\card{A}+k$.
Moreover, the free automata reduction contains polynomially-many (in $k$) steps.
Together, this yields membership in $\NPTIME$ and proves Theorem~\ref{Theorem:NP}.

\section{Complexity for Fixed Graphs}
\label{complexity}

We have seen that reachability under bounded context switching can
always be decided in $\NPTIME$, even if the graph describing the
storage mechanism is part of the input.  In this section, we study how
the complexity of the problem depends on the storage mechanism,
i.e. the graph. We fix the graph $G$ and consider the problem
$\BCSREACH{G}$.  We will see that for some graphs, the complexity is
lower than $\NPTIME$: We exhibit a class of graphs $G$ for which
$\BCSREACH{G}$ is solvable in polynomial time and we describe those
graphs for which the problem is $\NLOGSPACE$-complete.
Of course, for any graph $G$, the problem $\BCSREACH{G}$ is
$\NLOGSPACE$-hard, because reachability in directed graphs is. In some
cases, we also have an $\NLOGSPACE$ upper bound.

A loop-free graph is
a \emph{clique} if any two distinct vertices are adjacent. By $G^-$ we
denote the graph obtained from $G$ by removing all self-loops.  If
$G^-$ is a clique, then valence systems over $G$ are systems with
access to a fixed number of independent counters, some of which are blind and some
of which are partially blind.
\begin{theorem}\label{complexity:nl}
  If $G^-$ is a clique, then $\BCSREACH{G}$ is $\NLOGSPACE$-complete.
  Otherwise, $\BCSREACH{G}$ is $\PTIME$-hard.
\end{theorem}


\begin{figure}[b]
  {\centering
  \subcaptionbox{The graph $\Pfour$.}[4cm][c]{\picpfour{1}}}
  {\centering
  \subcaptionbox{The graph $\Cfour$.}[3cm][c]{\piccfour{1}}}
\caption{The graphs $\Pfour$ and $\Cfour$.}\label{figure:cfourpfour}
\end{figure}

In some cases, $\BCSREACH$ is $\PTIME$-complete. A loop-free graph is
a \emph{transitive forest} if it is obtained from the empty graph
using \emph{disjoint union} and \emph{adding a universal vertex}.
A universal vertex is a vertex that is adjacent to all other vertices.
Adding one means that we take a graph $G=(V,I)$ and add a new vertex
$v\notin V$ and make it adjacent to every vertex in $G$. Hence, we
obtain $(V\cup\{v\}, I\cup\{\{u,v\}\mid u\in V\})$.
%
\begin{theorem}\label{complexity:p}
If $G^-$ is a transitive forest, then $\BCSREACH{G}$ is in $\PTIME$.
\end{theorem}
In the area of graph monoids, transitive forests are an important
subclass. For many decision problems, they characterize those
graphs for which the problem becomes
decidable~\cite{Zetzsche2017a,LohreySteinberg2008} or
tractable~\cite{LohreyZetzsche2017b}. Intuitively, the storage
mechanisms represented by graphs $G$ where $G^-$ is a transitive
forest are those obtained by \emph{building stacks} and \emph{adding counters}, see~\cite{Zetzsche2017a,Zetzsche2016c}.

If $G=(V,I)$ is a graph, then $H$ is an \emph{induced subgraph of $G$}
if $H$ is isomorphic to a graph $(V',I')$, where $V'\subseteq V$ and
$I'=\{e\in I\mid e\subseteq V'\}$. See Fig.~\ref{figure:cfourpfour}
for the graphs $\Cfour$ and $\Pfour$.
\begin{restatable}{rtheorem}{restateTheoremComplexityCfourCompleteness}
\label{complexity:cfour:completeness}
  If $\Cfour$ is an induced subgraph of $G^-$, then
  $\BCSREACH{G}$ is $\NPTIME$-complete.
\end{restatable}
\noindent
It is an old combinatorial result that a simple graph is a transitive
forest if and only if it does not contain the two graphs $\Pfour$ and
$\Cfour$ as induced subgraphs~\cite{Wolk1965}. Hence, if one could
also show that $\BCSREACH{G}$ is $\NPTIME$-hard when $G^-=\Pfour$,
then Theorem~\ref{complexity:p} would cover all cases with polynomial
complexity (unless $\PTIME=\NPTIME$). However, we currently do not
know whether $\BCSREACH{\Pfour}$ is $\NPTIME$-hard.

\subparagraph*{Proof Sketches.}

The rest of this section is devoted to
sketching the proofs of Theorems~\ref{complexity:nl},
\ref{complexity:p}, and \ref{complexity:cfour:completeness}.  The
first step is a reformulation of the problem $\BCSREACH{G}$ if $G$ is
obtained from two disjoint graphs $G_0$ and $G_1$ by drawing edges
everywhere between $G_0$ and $G_1$. Suppose $G_i=(V_i,I_i)$ is a graph
for $i=0,1$ such that $V_0\cap V_1=\emptyset$. Then the graph
$G_0\times G_1$ is defined as $(V,I)$, where $V=V_0\cup V_1$ and
$I=I_0\cup I_1\cup \{\{v_0,v_1\}\mid v_0\in V_0, v_1\in V_1\}$.

The reformulation also involves valence automata, which can read
input.  Let $G=(V,I)$ be a graph and let
$\ops=\{o^+, o^-\mid o\in V\}$. A \emph{valence automaton} over $G$ is
a tuple $A=(Q,\Sigma,q_0,E,q_f)$, where $Q$ is a finite set of
\emph{states}, $\Sigma$ is an alphabet, $q_0\in Q$ is its
\emph{initial state},
$E\subseteq Q\times (\Sigma\cup\{\varepsilon\})\times (\ops\cup
\{\varepsilon\})\times Q$ is its set of \emph{transitions}, and
$q_f\in Q$ is its  \emph{final state}. A
\emph{configuration} is a tuple $(q,u,v)$, where $q\in Q$,
$u\in\Sigma^*$, and $v\in\ops^*$, where $v$ is right-invertible.
Intuitively, a transition $(q,s,w,q')$ changes the state from $q$ to $q'$,
reads the input $s$, and puts $w$ into the storage.  We write
$(q,u,v)\to (q',u',v')$ if there is a transition $(q,s,w,q')$ such
that $u'=us$ and $v'=vw$.
%
For any $k\in\N$,
the \emph{language accepted by $A$ with at most $k$ context switches} is
denoted $\langof[k]{A}$ and defined as the set of all $u\in\sigma^*$ such that from $(q_0,\varepsilon,\varepsilon)$,
we can reach $(q_f,u,w)$ for some $w\in\ops^*$ with $\grapheq{w}=\graphneutral$ and $\cs{w}\le k$.
The following problem will be used to reformulate $\BCSREACH{G\times H}$.
\begin{problem}
  \problemtitle{Intersection under bounded context switching}
  \problemshort{($\BCSI{G}{H}$)}
  \probleminput{Alphabet $\Sigma$, valence automata $A,B$ over graphs $G,H$, resp., \newline
      with
    input alphabet~$\Sigma$,
    and bounds $k,\ell,m$ in unary.}
  \problemquestion{Is the intersection $\langof[k]{A}\cap\langof[\ell]{B}\cap\Sigma^{\le m}$ non-empty?}
\end{problem}
We are now ready to state the reformulation, which is not difficult to prove.
\begin{restatable}{rproposition}{restatePropositionComplexityDirectEquivalence}
\label{complexity:direct:equivalence}
  If $G=G_0\times G_1$, then $\BCSREACH{G}$ is logspace-interreducible
  with $\BCSI{G_0}{G_1}$.
\end{restatable}
\noindent
We can use Proposition~\ref{complexity:direct:equivalence} to show
that adding a universal vertex does not change the complexity. 
%
\begin{restatable}{rproposition}{restatePropositionComplexityReductionAddcounter}
\label{complexity:reduction:addcounter}
  If $G$ has a universal vertex $v$, then $\BCSREACH{G}$ reduces to
  $\BCSREACH{G\setminus v}$ in logspace.
\end{restatable}
This can be deduced from
Proposition~\ref{complexity:direct:equivalence} as follows. If $v$ is
a universal vertex, then $G=(G\setminus v)\times H$, where $H$ is a
one-vertex graph. In this situation, a valence automaton over $H$ is
equivalent to a one-counter automaton (OCA). It is folklore that an
$n$-state OCA accepts a word of length $m$ if and only if it does so
with counter values at most $O((mn)^2)$~\cite{ChistikovCHPW16}. We can
thus compute in logspace a finite automaton for the language
$R=\langof[\ell]{B}\cap\Sigma^{\le m}$. This means, our instance of
$\BCSI{G\setminus v}{H}$ reduces to emptiness of
$\langof[k]{A}\cap R$. Using the automaton for $R$, this is easily
turned into an instance of $\BCSREACH{G\setminus v}$. Note that
Proposition~\ref{complexity:reduction:addcounter} yields the upper
bound of Theorem~\ref{complexity:nl}. The $\PTIME$-hardness follows
from $\PTIME$-hardness of reachability in pushdown automata.

The $\PTIME$ upper bound in Theorem~\ref{complexity:p} follows
 from Proposition~\ref{complexity:reduction:addcounter} and
the following.
\begin{restatable}{rproposition}{restatePropositionComplexityUnion}
\label{complexity:union}
  If $\BCSREACH{G_i}$ is in $\PTIME$ for $i=0,1$, then $\BCSREACH{G_0 \dotcup G_1}$
  is in $\PTIME$ as well.
\end{restatable}
\noindent
Proposition~\ref{complexity:union} is shown using a saturation
procedure similar to the one in Section~\ref{Section:Algorithm}. In
the latter, we shortcut paths that read two (complementary)
instructions. Here, in contrast, we find states $p,q$ between which
there is an arbitrarily long path that reads instructions $w$ over one
graph $G_i$ for $i=0,1$ such that $\grapheq{w}=\graphneutral$ and
$\cs{w}\le k$. Then, we add an $\varepsilon$-transition between $p$
and $q$.


Finally, let us comment on the $\NPTIME$-hardness in
Theorem~\ref{complexity:cfour:completeness}. If $G=\Cfour$, this is
the well-known $\NPTIME$-hardness of reachability under bounded
context switching. If $G$ contains self-loops, we employ
Proposition~\ref{complexity:direct:equivalence}: If $G^-=\Cfour$, then
$G=G_0\times G_1$ for some graphs where each $G_i$ contains two
non-adjacent vertices. In this case, it is known that that valence
automata over $G_i$ accept the same languages as those over
$G_i^-$~\cite{Zetzsche2017a,Zetzsche2016c}.  Therefore, the
formulation in terms of $\BCSI{G_0}{G_1}$ allows us to conclude
hardness.




\section{Conclusion}

We have shown that for every storage represented by a graph monoid,
reachability under bounded context switches ($\BCSREACH$) is decidable
in $\NPTIME$. To this end, we show that after some preprocessing in a
saturation procedure, any computation with bounded context switches
decomposes into quadratically many blocks. These blocks then cancel
and commute with each other so as to reduce to the identity
element. Thus, one can guess a decomposition into blocks and verify
the cancellation and commutation relations among them.

For the subclass of graph monoids whose underlying
simple graph is a transitive forest, we have provided a
polynomial-time algorithm (Theorem~\ref{complexity:p}).  However, we
leave open whether there are other graph monoids for which the problem
is in $\PTIME$.

One has $\NPTIME$-hardness in the case that the underlying simple
graph contains $\Cfour$ as an induced subgraph, which corresponds to
the classical case of bounded context switching in concurrent
recursive programs. Since transitive forests are precisely those
simple graphs that contain neither $\Cfour$ nor $\Pfour$ as induced
subgraphs~\cite{Wolk1965}, showing $\NPTIME$-hardness for $\Pfour$
would imply that Theorem~\ref{complexity:p} captures all graphs with
polynomial-time algorithms (unless $\PTIME=\NPTIME$). Unfortunately,
the known hardness techniques for problems involving graph groups or
Mazurkiewicz traces over
$\Pfour$~\cite{AalbersbergHoogeboom1989,LohreySteinberg2008,LohreyZetzsche2017b,Zetzsche2017a}
do not seem to apply.

Moreover, there is a variety of under-approximations for concurrent
recursive
programs~\cite{TorreMP07,AtigBH08,BreveglieriCCC96,TorreN11,EQR11,Atig14,TI0TP15a}. It
appears to be a promising direction for future research to study
generalizations of these under-approximations to valence systems.




\newpage
\bibliographystyle{plainurl}
\bibliography{12mmz}

\newpage
\appendix

\section{Proofs for Section~\ref{Section:Decomposition}}

\restateLemmaGedoens*

\begin{proof}
    Clearly, $w\rto^*\varepsilon$ implies $\grapheq{w}=\graphneutral$.
    We prove the converse using another rewriting relation that has been studied before~\cite{Zetzsche2016c,Zetzsche2013a}.
    Let $u\vdash v$ if either
    (i)~$u=s.\inc{o}.\dec{o}.t$ and $v=s.t$ for some $s,t \in\ops^*$ and  $o\in\ops$ or
    (ii)~$u=s.a.b.t$ and $v=s.b.a.t$ for some $s,t \in\ops^*$ and $a\in \incdec{o_1}$, $b\in \incdec{o_2}$ for some $o_1 \indeprel o_2$.
    Let $\vdash^*$ be the reflexive transitive closure of $\vdash$. It was shown in~\cite{Zetzsche2016c,Zetzsche2013a} that $\grapheq{w}=\graphneutral$ if and only if $w\vdash^*\varepsilon$.
    
    Let us show by induction on $|w|$ that $w\vdash^* \varepsilon$ implies $w\rto^*\varepsilon$.
    Consider a sequence of steps witnessing $w\vdash^*\varepsilon$.
    Let us call a step \emph{undesirable} if we cannot (directly) match it with $\rto$.
    This means, a step $u\vdash v$ where $u=s.a.b.t$ and $v=s.b.a.t$ with $\set{a,b}=\set{\inc{o},\dec{o}}$.
    
    If the sequence does not apply an undesirable step, it is already a reduction for $w$. 
    If such a step does occur, suppose  $w'$ is the first word where we apply one:
    We have $w\rto^* w'=s.a.b.t \vdash s.b.a.t\vdash^*\varepsilon$. Observe that then
    $\grapheq{a.b}=\graphneutral$ and hence
    $\grapheq{s.t}=\grapheq{s.a.b.t}=\graphneutral$. Since $|s.t|<|w'|\leq |w|$,
    induction yields $s.t \rto^*\varepsilon$ and thus $w\rto^* s.a.b.t\rto
    s.t \rto^*\varepsilon$.
\end{proof}

\restateLemmaUniqueFactorization*

\begin{proof}
    It suffices to show that each position in a context belongs to
    exactly one block.  Since $\pi$ reduces $w$ to $\varepsilon$, for
    each $x \in \oneto{n}$, there is exactly one $y \in \oneto{n}$, $x
    \neq y$ such that Rule~\rcancel\ or Rule~\rcancelrev\ is applied either to
    $w'.\ltr{w}{x}.\ltr{w}{y}.w''$ or to $w'.\ltr{w}{y}.\ltr{w}{x}.w''$.
    Consequently, each $\ltr{w}{x}$ belongs to at least one block.
    
    We also need to show that no position in $w$ belongs to more than
    one block.  Towards a contradiction, assume there are blocks $u,v$
    that overlap, \ie $u = r.s$, $v = s.t$.  Then there is some context
    $w_i$ of $w$ that we may write as $w_i = w_{i}'.r.s.t.w_{i}''$.
    As $u$ is a block, there is another context $w_j$ such that $u$
    cancels with an infix of $w_j$, \ie $w_j = w_{j}'.u'.w_{j}''$ with
    $u \cancelrel u'$.  By the definition of $\cancelrel$, we have $u' =
    s'.r'$ such that $s \cancelrel s'$ and $r \cancelrel r'$.
    
    Similarly, there is a context $w_{\bar{j}}$ containing infix $v'$
    which cancels $v$.  As $s'$ is the unique infix of $w$ such that the
    operations in $s$ cancel out with $s'$, we need to have $\bar{j} =
    j$, and we can write $w_j = w_j'.t'.s'.r'.w_j''$ where $s', r'$ are
    as before and $t \cancelrel t'$.
    Consequently, we have \( r.s.t \cancelrel t'.s'.r' \) which
    contradicts the maximality of the blocks $u$ and $v$.
\end{proof}

\restateLemmaMinimalNesting*

\begin{proof}
    We argue that $\ord{w}$ is transitive and antisymmetric.
    As the domain of $\ord{w}$ is finite, this is sufficient to guarantee that a minimal nesting exists:
    We may start with an arbitrary pair $s_1 \cancelrel s_2$ and iteratively pick smaller pairs as long as possible.
    
    For transitivity, note that $(s_1 \cancelrel s_2) \ord{w} (t_1 \cancelrel t_2)$ and $(t_1 \cancelrel t_2) \ord{w} (r_1 \cancelrel r_2)$ implies that we can write
    \(
    w = \upr{w}{1} . r_1 . \upr{w}{2} . t_1 . \upr{w}{3} . s_1 . \upr{w}{4} . s_2 . \upr{w}{5} . t_2 . \upr{w}{6} . r_2 . \upr{w}{7}
    \ ,
    \)
    proving $(s_1 \cancelrel s_2) \ord{w} (r_1 \cancelrel r_2)$.
    
    For antisymmetry, assume $(s_1 \cancelrel s_2) \ord{w} (t_1 \cancelrel t_2)$ and $(s_1 \cancelrel s_2) \ord{w} (t_1 \cancelrel t_2)$.
    This implies that we can write
    \(
    w = \upr{w}{1} . s_1 . \upr{w}{2} . t_1 . \upr{w}{3} . s_1 . \upr{w}{4} . s_2 . \upr{w}{5} . t_2 . \upr{w}{6} . s_1 . \upr{w}{7}
    \ ,
    \)
    a contradiction to the fact that $s_1$ has a unique occurrence in $w$.
\end{proof}

\restateLemmaCancelrelInverses*

\begin{proof}
    We proceed by induction on $\card{s_1} = \card{s_2}$.
    In the base case, $s_1 = a$ and $s_2 = b$ are single operations.
    If $a = \inc{o}$, $b = \dec{o}$ for some $o$, the statement obviously holds.
    Otherwise, we have $a = \dec{o}$, $b = \inc{o}$.
    By definition of $\rto$, this implies that $o \indeprel o$ holds, and $\grapheq{\dec{o}.\inc{o}} = \grapheq{\inc{o}.\dec{o}} = \graphneutral$ follows as desired.
    
    Assume that $s_1 = u.t$, $s_2 = r.v$ such that $u \cancelrel v$, $t \cancelrel r$.
    We may apply induction and use that $\cong$ is a congruence, obtaining
    \(
    \grapheq{s_1.s_2} = \grapheq{u.t.r.v} = \grapheq{u.v} = \graphneutral
    \ .
    \)
\end{proof}

\restateLemmaCancelrelNotSameBlock* 

\begin{proof}
    We show that a context cannot contain two operations that cancel out.
    As two blocks that cancel out would contain such operations, this is sufficient.
    
    Towards a contradiction, assume $a \cancelrel b$ where $a,b$ are contained in the same context $w_i$,
    \ie we have
    \(
    w_i = w_{i_1} . a . w_{i_2} . b . w_{i_3}
    \ .
    \)
    We have $a = \inc{o}, b = \dec{o}$, or $a = \dec{o}, b = \inc{o}$ and $o \indeprel o$.
    If $w_{i_2} = \varepsilon$, we obtain a contradiction to the assumption that $w_i$ is irreducible in both cases.
    If $w_{i_2}$ is a sequence of operations from $\incdec{o}$, we also obtain a contradiction to irreducibility.
    Otherwise $w_i$ contains some operation $c \in \incdec{o_2}$ for $o_2 \neq o$.
    Since $w_i$ is a context, $o_2$ and $o$ are not independent.
    Since $a$ should cancel with $b$, $\pi$ needs to swap one of them over $c$, or it needs to swap the inverse of $c$ (which is also in $\incdec{o_2}$) over one of them.
    As $o_2$ and $o$ are not independent, this is not possible.
    We obtain that $a$ cannot cancel with $b$, a contradiction.
\end{proof}

\section{Proofs for Section~\ref{Section:Algorithm}}

\restateLemmaSyntacticInverse*

\begin{proof}
    We show that for dependent and irreducible $u,v$, $\grapheq{u}\cdot\grapheq{v}=\graphneutral$ implies $v=\syninv{u}$.
    
    We proceed by induction on the length of $u$.
    In the base case, we have $u = \varepsilon$, which implies $\grapheq{v} = \grapheq{\varepsilon}$.
    As $v$ is irreducible and dependent, we have $v = \varepsilon$ as required.
    Here, we have used that any word that reduces to $\graphneutral$ needs to have a first reduction step in which canceling occurs, which can only exist if the word is non-irreducible.
    
    Assume that $u = u'.a$.
    We claim that we can write $v = b.v'$, where $b = \syninv{a}$ and $v' = \syninv{u'}$, which implies $v = \syninv{u}$.
    As we have $\grapheq{u}\cdot\grapheq{v}=\graphneutral$, $v$ contains an operation canceling $a$.
    If this operation is not the very first letter in $v$, we obtain a contradiction.
    
    Let $a \in \incdec{o_1}$ and assume that the first operation in $v$ is in $\incdec{o_2}$ for $o_1 \neq o_2$.
    Since $v$ is dependent, the first operation cannot commute with the inverse of $a$, a contradiction to $\grapheq{u}\cdot\grapheq{v}=\graphneutral$.
    Hence, $v$ starts with a prefix using operations in $\incdec{o_2}$ and containing the operation that cancels $a$.
    
    If $o_1 \indeprel o_1$ does not hold, then we need to have $a = \inc{o_1}$.
    If $a = \dec{o_1}$, we would have that $v$ is not right-invertible (since it is irreducible), a contradiction to the assumption $\grapheq{u}\cdot\grapheq{v}=\graphneutral$.
    Having $a = \inc{o_1}$ implies that the first operation $b$ in $v$ is $\dec{o_1}$, which is indeed the syntactic inverse of $a$.
    
    If $o_1 \indeprel o_1$ holds, then we have to consider both cases $a = \dec{o_1}$ and $a = \inc{o_1}$.
    In the first case, we claim that $b = \inc{o_1}$ has to hold.
    If $v$ starts with a sequence of $\dec{o_1}$, and then has an occurrence of $\inc{o_1}$, we get a contradiction to the irreducibility of $v$.
    Similarly, in the second case $a = \inc{o_1}$, $b = \dec{o_1}$ has to hold.
    
    Altogether, we have $v = b.v'$ with $b = \syninv{a}$.
    We have
    \[
        \graphneutral
        = \grapheq{u.v}
        = \grapheq{u'.a.b.v'}
        = \grapheq{u'} \cdot \grapheq{a.b} \cdot \grapheq{v'}
        = \grapheq{u'} \cdot \graphneutral \cdot  \grapheq{v'}
        = \grapheq{u'.v'}
        \ .
    \]
    Since $u'$ and $v'$ are still dependent and irreducible, we obtain $v' = \syninv{u'}$ by induction.
    We conclude $\syninv{u} = \syninv{u'.a} = \syninv{a}.\syninv{u'} = b.v' = v$ as desired.
\end{proof}

\section{Proofs for Section~\ref{complexity}}
We begin with the proof of
Proposition~\ref{complexity:direct:equivalence}, which consists of
three lemmas.

\begin{lemma}\label{complexity:reach:inter}
If $G=G_0\times G_1$, then $\BCSREACH{G}$ reduces to $\BCSI{G_0}{G_1}$ in logspace.
\end{lemma}

\begin{proof}
    Suppose $S=(Q,\to)$ is a valence system over the graph
    $G=G_0\times G_1$ and we are given the context switching bound $k$
    and states $q_{init}$ and $q_{fin}$. Let $G_i=(V_i,E_i)$ and
    $\ops_i=\{o^+,o^-\mid o\in V_i\}$ for $i=0,1$.
   Let $\Sigma=Q\times[0,k]\times Q$. Let  $q_0=q_{init}$, $p_{n+1}=q_{fin}$.
 
   The idea is to construct $A$ and $B$ so that a word
   \[ (p_1,s_1,q_1)(p_2,s_2,q_2)\cdots (p_n,s_n,q_n) \]
    in the intersection $\langof[k]{A}\cap\langof[k]{B}\cap\Sigma^{\le k}$
    witnesses a computation
    \[ q_0\tow{u_0} p_1 \tow{v_1}q_1\tow{u_1}  \cdots \tow{u_{n-1}} p_n\tow{v_n} q_n\tow{u_n}p_{n+1} \]
    in $S$ where $u_0,u_n\in\ops_0^*$ and $u_i\in\ops_0^+$ for $i\in[1,n-1]$ and $v_i\in\ops_1^+$ for $i\in[1,n]$
    and $\cs{v_i}=s_i$ for $i\in[1,n]$ and $\cs{u_0v_1u_1\cdots v_nu_n}\le k$. One checks easily
    that in this case, we have
    \[ \cs{u_0v_1u_1\cdots v_nu_n}=\cs{u_0}+\sum_{i=1}^n
      (\cs{v_i}+\cs{u_i}+2)=\cs{u_0}+\sum_{i=1}^n (s_i+\cs{u_i}+2). \]
    Note that if we can construct $A$ and $B$ in logspace so that a
    word in the intersection
    $\langof[k]{A}\cap\langof[k]{B}\cap\Sigma^{\le k}$ exists if and
    only if there is a computation as above, then we have indeed a
    logspace reduction from $\BCSREACH{G}$ to $\BCSI{G_0}{G_1}$.
    
    We accomplish this by constructing in logspace valence automata
    $A$ and $B$ over $G_0$ and $G_1$, respectively, for which
    \begin{multline}\label{complexity:splitlemma:a}
      \langof[k]{A}=\Bigg\{ (p_1,s_1,q_1)(p_2,s_2,q_2)\cdots(p_{n},s_n,q_n)\mid\text{for each $i\in[0,n]$, we have $q_i\tow{u_i}p_{i+1}$} \\
      \text{with $u_0,u_n\in\ops_0^*$, $u_1,\ldots,u_{n-1}\in\ops_0^+$ and $\cs{u_0}+\sum_{i=1}^n (s_i+\cs{u_i}+2)\le k$} \Bigg\}
    \end{multline}
    and
    \begin{multline}\label{complexity:splitlemma:b}
      \langof[k]{B}=\Bigg\{(p_1,s_1,q_1)(p_2,s_2,q_2)\cdots(p_n,s_n,q_n) \mid \text{for each $i\in[1,n]$, we have $p_i\tow{v_i} q_i$} \\
      \text{for some $v_i\in\ops_1^+$ with $\cs{v_i}=s_i$ and $\cs{v_1\cdots v_n}\le k$} \Bigg\}.
    \end{multline}
    It is not difficult to construct $A$ and $B$. The state set of $A$
    is $Q_A=Q\times 2^{\ops_0}\times[-1,k]$. In its right-most
    component, it counts the number of context switches observed
    during the run including the ones contributed by edges reading
    letters in $\Sigma$.  In order to update this counter, it used its
    middle component. Here, it keeps track of the set of operations
    seen in the current context. Note that the state set of $A$ has
    polynomial size because we treat the graph $G$ as constant.
    The automaton $A$ has edges for simulating runs $u_i$ in $S$
    that are labeled $\varepsilon$. Moreover, for each
    $(p,s,q)\in\Sigma$, $A$ has an edge that changes the left
    component from $p$ to $q$ and adds $s$ to the counter.

    The state set of $B$ is
    $\{\ast\}\cup Q\times Q\times 2^{\ops_1}\times[0,k]$.  Initially,
    $B$ is in the state $\ast$. For each $(p,s,q)\in\Sigma$ and
    $p\tow{x}p'$ for $x\in\ops_1$, $B$ has an edge labeled $(p,s,q)$
    from $\ast$ to $(p',q,\{x\},s)$. In states from
    $Q\times Q\times 2^{\ops_1}\times[0,k]$, $B$ operates analogous to
    $A$, except that the second $Q$-component remains constant and the
    right-most component counts downward. On these edges, $B$ reads no
    input. In addition, from a state $(q,q,T,0)$ with
    $T\subseteq \ops_1$, $B$ can go back to the state $\ast$, which is
    also its final state.

    With $A$ and $B$ set up this way, it is clear that
    \eqref{complexity:splitlemma:a} and \eqref{complexity:splitlemma:b}
    are satisfied. Thus we have indeed that $q_{init}$ can reach
    $q_{fin}$ with at most $k$ context switches if and only if
    $\langof[k]{A}\cap\langof[k]{B}\cap\Sigma^{\le k}\ne\emptyset$.
\end{proof}

In the rest of this section, we will also use the notation
$\langof{A}=\bigcup_{k\ge 0} \langof[k]{A}$ for valence automata
$A$.

\begin{lemma}\label{complexity:builtincounter}
    Given a valence automaton $A$ and a unary bound $k$, one can
    construct in logspace a valence automaton $A'$ with
    $\langof{A'}=\langof[k]{A'}=\langof[k]{A}$.
\end{lemma}

\begin{proof}
    Let $A=(Q,\Sigma,q_0,E,q_f)$. Then, $A'$ has states
    $\{\ast\}\cup Q\times 2^{\ops}\times[0,k]$. It simulates
    computations of $A$. In the third component it counts the
    number of context switches it sees on the storage. In order to
    maintain this counter, it stores in the second component the set
    of operations in $\ops$ occurring in the current context. The
    initial state of $A'$ is $(q_0,\emptyset,0)$ and the final state
    is $\ast$. In order to reach the final state, $A'$ has a
    transition that reads input $\varepsilon$ and adds $\varepsilon$
    to the storage from every state $(q_f,U,\ell)$ to $\ast$.
\end{proof}

\begin{lemma}\label{complexity:inter:reach}
    If $G=G_0\times G_1$, then $\BCSI{G_0}{G_1}$ reduces to $\BCSREACH{G}$ in logspace.
\end{lemma}

\begin{proof}
    Suppose we are given an alphabet $\Sigma$, valence automata
    $A=(Q_A,\Sigma,q_{0,A},E_A,q_{f,A})$ and $B=(Q_B,\Sigma,q_{0,B},E_B,q_{f,B})$ over
    $G_0$ and $G_1$, respectively, such that
    $\langof{A},\langof{B}\subseteq\Sigma^*$ and unary bounds
    $k,\ell,m\in\N$. According to
    Lemma~\ref{complexity:builtincounter}, we may assume that
    $\langof{A}=\langof[k]{A}$ and $\langof{B}=\langof[\ell]{B}$. Let
    us call a sequence of transitions in a valence automaton an
    \emph{$a$-path} if it reads $a\in\Sigma$ from the input. Hence, it
    consists of some $\varepsilon$-transitions, an $a$-transition, and
    again some $\varepsilon$-transitions.
    
    We construct a valence automaton $C$ over $G$ as follows.  The
    state set of $C$ is a product of several components, among them
    the state sets $Q_A$ and $Q_B$. When $C$ reads a letter
    $a\in\Sigma$, it first simulates an $a$-path of $A$ and then an
    $a$-path of $B$. Moreover, it has a counter that counts the number
    of read input symbols and makes sure that at most $m$ of them are
    read. It is obvious how to set up the state set and transitions of
    $C$ appropriately so that for some states $q_{init}$ and $q_{fin}$
    in $C$, we have: $q_{init}\tow{w}q_{fin}$ with
    $\grapheq{w}=\graphneutral$ if and only if
    $\langof{A}\cap\langof{B}\cap\Sigma^{\le
      m}\ne\emptyset$. Moreover, since we know that
    $\langof{A}=\langof[k]{A}$ and $\langof{B}=\langof[\ell]{B}$, we
    have that $\langof{A}\cap\langof{B}\cap\Sigma^{\le m}\ne\emptyset$
    implies
    $\langof[k]{A}\cap\langof[\ell]{B}\cap\Sigma^{\le m}\ne\emptyset$
    and thus our construction yields $q_{init}\tow{w}q_{fin}$ with
    $\grapheq{w}=\graphneutral$ and $\cs{w}\le k+\ell+2m$.
\end{proof}
\noindent
Proposition~\ref{complexity:direct:equivalence} now follows from
Lemmas~\ref{complexity:reach:inter} and
\ref{complexity:inter:reach}.
  
\restatePropositionComplexityDirectEquivalence*

\begin{proof}
    Lemmas~\ref{complexity:reach:inter} and
    \ref{complexity:inter:reach} show the two reductions.
\end{proof}

We are now ready to prove Proposition~\ref{complexity:reduction:addcounter}.

\restatePropositionComplexityReductionAddcounter*

\begin{proof}
    Since $v$ is universal, we may apply
    Lemma~\ref{complexity:reach:inter} in the case where
    $G_0=G\setminus v$ and $G_1$ is the subgraph of $G$ induced by
    $v$. Hence, it suffices to show that $\BCSI{G_0}{G_1}$ reduces to
    $\BCSREACH{G\setminus v}$. Therefore, suppose we are given an
    alphabet $\Sigma$, valence automata $A$ and $B$ over $G_0$ and
    $G_1$, respectively, and bounds $k,\ell,m$.
    
    The rest of our proof employs a classical infinite-state model.  A
    \emph{one-counter automaton} is a tuple $C=(Q,\Gamma,q_0,E,q_f)$,
    where $Q$ is a finite set of \emph{states}, $\Gamma$ is an
    alphabet, $q_0\in Q$ is its \emph{initial state},
    $q_f\in Q$ is its \emph{final state}, and
    $E\subseteq Q\times(\Gamma\cup\{\varepsilon\})\times
    \{-1,1,0,=0\}\times Q$ is its set of \emph{transitions}. A
    \emph{configuration} is a pair $(q,n)\in Q\times\N$. The
    transition relation $(p,m)\tow{w}(q,n)$ between configurations
    $(p,m)$ and $(q,n)$ is defined as expected for $w\in\Gamma^*$.
    Here, $-1$, $1$, $0$, $=0$ stand for \emph{decrement},
    \emph{increment}, \emph{no change}, and \emph{zero test},
    respectively.  A word $w\in\Gamma^*$ is \emph{accepted} by the
    automaton if $(q_0,0)\tow{w}(q_f,n)$ for some
    $\langof{C}$.
    
    We construct in two steps a one-counter automaton $C$ with
    $\langof{C}=\langof[\ell]{B}$.  First, we use Lemma~\ref{complexity:builtincounter}
    to construct a valence automaton $B'$ with $\langof{B'}=\langof[\ell]{B}$.
    Now we can construct $C$. Since $G_1$ contains only one vertex, it
    is easy to construct in logspace a one-counter automaton $C$ with
    $\langof{C}=\langof{B'}$: If $v$ has a loop, then a valence
    automaton over $G_1$ is an automaton with access to a counter that
    assumes values in $\Z$ and has an increment and a decrement
    operation. If $v$ has no loop, then a valence automaton over $G_1$
    is an automaton with a counter that assumes values in $\N$ and,
    again, has an increment and a decrement operation. In each case,
    we can easily simulate the monoid using a counter that assumes
    only values in $\N$ and has zero tests available.
    
    Thus, we have reduced our problem to the question of whether
    $\langof[k]{A}\cap \langof{C}\cap \Sigma^{\le m}\ne\emptyset$.  It
    is folklore that if a one-counter automaton with $n$ states
    accepts any word, then it does so in a computation during which
    the counter values are bounded by $O(n^2)$
    (see~\cite{ChistikovCHPW16} and the references therein). This
    implies that if $C$ has $n$ states and accepts a word in
    $\Sigma^{\le m}$, then it does so in some computation with counter
    values at most $a(mn)^2+b$ for some constants $a,b\in\N$.
    
    Let us constuct a finite automaton $D$ that simulates all
    computations of $C$ where the counter does not exceed $a(mn)^2+b$
    and which read an input word of length $\le m$. Then, $D$ has
    $(a(mn)^2+b)nm$ states and can clearly be constructed in
    logspace. Moreover, we have
    $\langof{D}=\langof{C}\cap\Sigma^{\le m}$ and thus
    $\langof[k]{A}\cap\langof[\ell]{B}\cap\Sigma^{\le m}\ne\emptyset$
    if and only if $\langof[k]{A}\cap L(D)\ne\emptyset$. Using a
    simple product construction, we can now obtain a valence automaton
    $A'$ from $A$ and $D$ so that
    $\langof[k]{A'}=\langof[k]{A}\cap\langof{D}$. Finally, checking
    emptiness of $\langof[k]{A'}$ is an instance of $\BCSREACH{G_0}$.
\end{proof}

Our next goal is to prove Proposition~\ref{complexity:union}. For
the proof, it will be convenient to use a reformulation of the
problem $\BCSREACH$.  A valence system $S=(Q,\to)$ is said to be
\emph{$k$-bounded} if $p\tow{w} q$ for states $p,q\in Q$ implies
$\cs{w}\le k$.

\begin{problem}
    \problemtitle{BCS reachability with promise}
    \problemshort{($\BCSRP{G}$)}
    \probleminput{Bound $k$ in unary, a $k$-bounded valence system $S$, states $q_{init}$, $q_{fin}$.}
    \problemquestion{Is there a run from $(\qinit,\varepsilon)$ to $(\qfin, w)$ so that $\grapheq{w} = \graphneutral$?}
\end{problem}

\noindent
Of course, the problem $\BCSRP{G}$ is very similar to
$\BCSREACH{G}$. What makes it useful is that in order to solve
$\BCSRP{G}$ we need not worry about discovering paths with too many
context switches: We may assume that all paths are guaranteed to
contain at most $k$. In particular, our procedure will find paths in
a valence system over $G=G_0\dotcup G_1$ by repeatedly finding paths
in valence systems over $G_0$ or $G_1$. The formulation in terms of
$\BCSRP{G}$ relieves us from keeping track of how many context
switches we have made in the paths for the $G_i$ when composing them
to paths over $G$.

\begin{lemma}
\label{complexity:promise}
    For each graph $G$, the problems $\BCSREACH{G}$ and $\BCSRP{G}$
    are inter-reducible via logspace reductions.
\end{lemma}

\begin{proof}
    Of course, $\BCSRP{G}$ reduces trivially to
    $\BCSREACH{G}$. Conversely, suppose $S$ is a valence system over
    $G$ and we want to decide whether there is a run from
    $(q_{init},\varepsilon)$ to $(q_{fin},w)$ so that
    $\grapheq{w}=\graphneutral$ and $\cs{w}\le k$.
    
    We turn $S$ into a valence system $S'$ that keeps track of the
    number of context switches in its states and is therefore
    $k$-bounded. If $G=(V,I)$, $\ops=\{o^+, o^-\mid o\in V\}$,
    $S=(Q,\to)$, then $S'$ has states $Q'=\{\ast\}\cup Q\times 2^\ops\times [0,k]$.
    In its middle component, it maintains the list of operations
    occurring in the current context. This is used to update the
    right-most component, which counts the number of context switches
    $S'$ has seen during the computation.  Hence, we have transitions:
    \begin{align*}
      &(p,U,\ell)\tow{x}(q,U\cup\{x\},\ell)  && \text{for each $p\tow{x}q$ in $S$ where $U\cup\{x\}$ is dependent}\\
      &                                      && \text{and $\ell\in[0,k]$} \\
      &(p,U,\ell)\tow{x}(q,\{x\},\ell+1)       && \text{for each $p\tow{x}q$ in $S$ where $U\cup\{x\}$ is not dependent} \\
      &                                      && \text{and $\ell\in[0,k-1]$} \\
      &(p,U,\ell)\tow{\varepsilon}(q,U,\ell) && \text{for each $p\tow{\varepsilon} q$ in $S$}
    \end{align*}
    Now reachability of $(q_{fin},w)$ from $(q_{init},\varepsilon)$
    with $\grapheq{w}=\graphneutral$ and $\cs{w}\le k$ is equivalent to
    reachability of some $((q_{fin},U,\ell),w)$ from
    $((q_{init},\emptyset,0),\varepsilon)$ with
    $\grapheq{w}=\graphneutral$. Since we want to provide a single
    final state in our reduction, we also add transitions
    \begin{align*}
      &(q_{fin},U,\ell)\tow{\varepsilon}\ast && \text{for each $U\subseteq \ops$ and $\ell\in[0,k]$}
    \end{align*}
    Then $S'$ is clearly $k$-bounded. As initial state, we take
    $q'_{init}=(q_{init},\emptyset,0)$ and as a final state, we take
    $q'_{fin}=\ast$. Then we clearly have a run from
    $(q'_{init},\varepsilon)$ to $(q'_{fin},w')$ with
    $\grapheq{w'}=\graphneutral$ if and only if there is a run from
    $(q_{init},\varepsilon)$ to $(q_{fin},w)$ with
    $\grapheq{w}=\graphneutral$ and $\cs{w}\le k$.
\end{proof}

We are now ready to prove Proposition~\ref{complexity:union}.
  
\restatePropositionComplexityUnion*

\begin{proof}
    According to Lemma~\ref{complexity:promise}, it suffices to show
    that if $\BCSRP{G_0}$ and $\BCSRP{G_1}$ are in $\PTIME$, then
    $\BCSRP{G_0\dotcup G_1}$ is in $\PTIME$. To this end, we use a
    saturation algorithm.
    
    Suppose we are given a $k$-bounded valence system $S=(Q,\to)$ over
    $G=G_0 \dotcup G_1$. Let $G_i=(V_i,I_i)$ and
    $\ops_i=\{o^+, o^-\mid o\in V\}$ for $i=0,1$. We add edges labeled
    $\varepsilon$ according to the following rule: If from
    $(p,\varepsilon)$, one can reach $(q,w)$ with
    $\grapheq{w}=\graphneutral$ and $w\in\ops_i^*$ for some
    $i\in\{0,1\}$, then we add an edge $p\tow{\varepsilon}q$. Note
    that we can decide in polynomial time whether this is the case:
    Restricting $S$ to edges with labels in $\ops_i^*$ yields a
    valence system hat inherits $k$-boundedness from $S$; thus, we
    answer an instance of $\BCSRP{G_i}$. Moreover, adding an
    $\varepsilon$-transition preserves $k$-boundedness as
    well. Finally, since we are only adding transitions labeled
    $\varepsilon$, this procedure terminates after adding at most
    $|Q|^2$ transitions and thus in polynomial time. Let
    $S'=(Q,\leadsto)$ be the resulting valence system.
    
    We claim that $q_{init}\leadstow{\varepsilon}q_{fin}$ if and only
    if $q_{init}\tow{w}q_{fin}$ for some $w$ with
    $\grapheq{w}=\graphneutral$.
    Here, the ``only if'' direction
    follows easily by induction on the number of steps performed in
    the saturation algorithm.
    
    For the converse, we prove by induction on $|w|$ that for any
    $p,q\in Q$, if $p\leadstow{w}q$ with $\grapheq{w}=\graphneutral$,
    then $p\leadstow{\varepsilon}q$.  Suppose we have $p\leadstow{w}q$
    with $\grapheq{w}=\graphneutral$.
    
    If $|w|=0$, we are done. Moreover, if $w\in\ops_i^+$ for some
    $i\in\{0,1\}$, then we have introduced $p\leadstow{\varepsilon} q$
    during the saturation.  Otherwise, we can decompose
    $w=u_1\cdots u_n$ so that $u_j\in\ops_0^+\cup\ops_1^+$ and
    $u_j\in\ops_i^+$ iff $u_{j+1}\in\ops_{1-i}^+$ for $j\in[1,n]$.

    Then it follows from the definition of $\cong$ that
    there is some $\ell\in[1,n]$ with $u_j\cong\varepsilon$:
    Otherwise, it follows
    by induction on the number of applied equivalences (i.e.
    $o_1^{\pm}.o_2^{\pm}\cong o_2^{\pm}.o_1^{\pm}$ for $o_1 I o_2$ or
    $o^{+}o^{-}\cong \varepsilon$), that every word $w'$ with
    $w\cong w'$ has a decomposition $w'=u'_1\cdots u'_n$ with $u'_j\cong u_j$ for $j\in[1,n]$.
    Hence, let $u_\ell\cong \varepsilon$ and $w=xu_\ell y$.
    
    Let $p',q'$ be states so that
    $p\leadstow{x}p'\leadstow{u_\ell}q'\leadstow{y}q$. Since
    $u_\ell\cong\varepsilon$, the saturation procedure has added the
    transition $p'\leadstow{\varepsilon} q'$. Hence, we have
    $p\leadstow{xy}q$ with
    $\grapheq{xy}=\grapheq{xu_\ell y}=\grapheq{w}=\graphneutral$ and
    thus by induction $p\leadstow{\varepsilon} q$.
\end{proof}

\restateTheoremComplexityCfourCompleteness*
  
\begin{proof}
    Since we have shown $\NPTIME$ membership for all graphs, it
    suffices to show $\NPTIME$-hardness in the case $G^-=\Cfour$. In
    that case, we have $G=G_0\times G_1$, where each $G_i$ consists of
    two vertices, which may or may not carry self-loops.
    
    According to Proposition~\ref{complexity:direct:equivalence}, it
    suffices to show that $\BCSI{G_0}{G_1}$ is
    $\NPTIME$-hard. Consider a 3CNF-SAT instance
    $\varphi=\bigwedge_{j=1}^m C_j$, where each $C_j$ is a clause over
    the variables $\{x_1,\ldots,x_n\}$. We encode an assignment of the
    variables by a word $w\in \{0,1\}^n$ and say that $w$
    \emph{satisfies $C_j$ (resp. $\varphi$)} if the corresponding
    assignment satisfies $C_j$ (resp. $\varphi$).
    
    \newcommand{\rev}[1]{{#1}^{\mathsf{rev}}}
    Since each $G_i$ consists of two non-adjacent vertices, it is
    known that valence automata over $G_i$ accept exactly the
    context-free
    languages~\cite{Zetzsche2017a,Zetzsche2016c}. Moreover, the
    translation from PDAs (or context-free grammars) into valence
    automata over $G_i$ can obviously performed in logarithmic space.
    For a word $w\in\{0,1\}^*$, let $\rev{w}$ denote its reversal.
    Consider the context-free languages
    \begin{align*}
      \calL_0&=\{w_1\#\rev{w_1} \# w_2 \# \rev{w_2} \cdots w_m \# \rev{w_m}\# \mid \text{for $j\in[1,m]$, $w_j\in\{0,1\}^n$ satisfies $C_j$} \} \\
      \calL_1&=\{w_0\# w_1 \# \rev{w_1}\# \cdots w_{m-1} \# \rev{w_{m-1}} \# w_m\# \mid \text{for $j\in[0,m]$, $w_j\in\{0,1\}^n$} \}
    \end{align*}
    over $\Sigma=\{0,1,\#\}$.  On the one hand, we have
    $\calL_0\cap \calL_1\cap\Sigma^{2m(n+1)}\ne\emptyset$ if and only if
    $\varphi$ is satisfiable. On the other hand, we can construct a
    PDA for each $\calL_i$ in logarthmic space, and hence also valence
    automata $A_i$ over $G_i$ such that $\langof{A_i}=\calL_i$ for
    $i=0,1$. Since in $G_i$, there are no edges between distinct
    vertices, every computation has $0$ context switches, meaning
    $\langof{A_i}=\langof[0]{A_i}$. Thus, $\varphi$ is satisfiable if
    and only if
    $\langof[0]{A_0}\cap\langof[0]{A_1}\cap\Sigma^{2m(n+1)}\ne\emptyset$.
\end{proof}


\end{document}